\documentclass[twocolumn,10pt]{IEEEtran}
\normalsize
\usepackage{bbm}
\usepackage{adjustbox}
\usepackage{enumitem}
\usepackage{cite,algorithm,algorithmic,amsmath,amssymb,amsthm,empheq,mhsetup}
\usepackage{subfigure,amsfonts,balance}
\usepackage{epstopdf}
\usepackage{setspace}
\usepackage[dvipsnames]{xcolor}

\DeclareMathOperator{\EEE}{\mathbb{E}}

\DeclareMathOperator{\C}{\mathbb{C}}

\DeclareMathOperator{\f}{\pmb{f}}
\DeclareMathOperator{\aaa}{\pmb{a}}
\DeclareMathOperator{\FF}{\mathcal{F}}

\DeclareMathOperator{\OO}{\mathcal{O}}

\DeclareMathOperator{\vv}{\pmb{v}}

\DeclareMathOperator{\HHH}{\mathcal{H}}

\DeclareMathOperator{\RRR}{\mathbb{R}}

\DeclareMathOperator{\LL}{\mathcal{L}}

\DeclareMathOperator{\CN}{\mathcal{CN}}

\DeclareMathOperator{\MM}{\mathcal{M}}

\DeclareMathOperator{\NN}{\mathcal{N}}

\DeclareMathOperator{\rr}{\pmb{r}}
\DeclareMathOperator{\x}{\pmb{x}}
\DeclareMathOperator{\ttt}{\pmb{t}}

\DeclareMathOperator{\uu}{\pmb{u}}

\DeclareMathOperator{\g}{\pmb{g}}

\DeclareMathOperator{\ETA}{\pmb{\eta}}

\DeclareMathOperator{\VARPHI}{\pmb{\varphi}}

\DeclareMathOperator{\ZETA}{\pmb{\zeta}}

\newtheorem{definition}{Definition}

\theoremstyle{remark}
\newtheorem{remark}{Remark}
\newtheorem{proposition}{Proposition}

\ifCLASSINFOpdf
\else
\fi

\allowdisplaybreaks
\begin{document}
\bstctlcite{IEEEexample:BSTcontrol}
%
\title{\huge  Joint Resource Allocation to Minimize Execution Time of Federated Learning in Cell-Free Massive MIMO}
%
%
\author{Tung~T.~Vu,~\IEEEmembership{Member,~IEEE}, Duy~T.~Ngo,~\IEEEmembership{Senior Member,~IEEE}, 
Hien~Quoc~Ngo,~\IEEEmembership{Senior Member,~IEEE}, 
Minh~N.~Dao, 
Nguyen~H.~Tran,~\IEEEmembership{Senior Member,~IEEE},
        and Richard~H.~Middleton,~\IEEEmembership{Fellow,~IEEE}
\thanks{T.~T.~Vu and H.~Q.~Ngo are with the Institute of Electronics, Communications, and Information Technology (ECIT), Queen's University Belfast, Belfast BT7 1NN, United Kingdom (e-mail: \{t.vu,hien.ngo\}@qub.ac.uk).}
\thanks{D.~T.~Ngo and R.~H.~Middleton are with the School of Engineering, The University of Newcastle, Callaghan, NSW 2308, Australia (e-mail: \{duy.ngo, richard.middleton\}@newcastle.edu.au).}
\thanks{N.~H.~Tran is with the School of Computer Science, The University of Sydney, Sydney, NSW 2006, Australia  (e-mail: nguyen.tran@sydney.edu.au).}
\thanks{M.~N.~Dao is with the School of Engineering, Information Technology and Physical Sciences, Federation University, Ballarat, VIC 3353, Australia (e-mail: m.dao@federation.edu.au).}
\vspace{-5mm}}

\maketitle
\vspace{-0mm}
\begin{abstract}
Due to its communication efficiency and privacy-preserving capability, federated learning (FL) has emerged as a promising framework for machine learning in 5G-and-beyond wireless networks. Of great interest is the design and optimization of new wireless network structures that support stable and fast operation of FL. Cell-free massive multiple-input multiple-output (CFmMIMO) turns out to be a suitable candidate, which allows each communication round in the iterative FL process to be stably executed within a large-scale coherence time. Aiming to reduce the total execution time of the FL process in CFmMIMO, this paper proposes choosing only a subset of available users to participate in FL. An optimal selection of users with favorable link conditions would minimize the execution time of each communication round, while limiting the total number of communication rounds required. Toward this end, we formulate a joint optimization problem of user selection, transmit power, and processing frequency, subject to a predefined minimum number of participating users to guarantee the quality of learning. We then develop a new algorithm that is proven to converge to the neighbourhood of the stationary points of the formulated problem. Numerical results confirm that our proposed approach significantly reduces the FL total execution time over baseline schemes. The time reduction is more pronounced when the density of access point deployments is moderately low.
\end{abstract}
\vspace{-0mm}
\begin{IEEEkeywords}
\vspace{-0mm}
Cell-free massive MIMO, federated learning, execution time minimization.
\vspace{-0mm}
\end{IEEEkeywords}

%
\IEEEpeerreviewmaketitle

\vspace{-3mm}
\section{Introduction}
\vspace{-0mm}
\label{sec:Introd}
The numbers of mobile devices and connections have been growing significantly in recent years. According to Cisco \cite{cisco20}, while the number of global mobile devices is expected to reach $13.1$ billion by 2023, more than $10\%$ of this figure will have 5G connections. These devices generate a vast amount of data, which in turn enable a wide range of on-device artificial intelligence (AI) services, such as traffic navigation, indoor localization, image recognition, natural language processing, and augmented reality \cite{dong19TMC,letaief19CM,cala18CM,viet21A}. However, it is impractical to use conventional centralized approaches to train AI models (especially those by deep neural networks) at mobile devices. As such approaches store and process data at distant cloud centers, they find it extremely challenging to support delay-critical applications. Moreover, uploading user raw data to distant cloud servers also raises serious concerns about user's data privacy \cite{zhu20CM}.

{Federated learning (FL) has recently emerged as a promising solution for AI model training at wireless devices with a certain guarantee of data privacy \cite{khan21CST,niknam20CM,li20SPM,du20OJCS,song20CIM,chen20CM}. An FL process is iterative and involves several communication rounds. In each communication round, users (UEs) compute their local model updates by using their local training data, followed by sending these updates to a central server. The central server  aggregates the received local updates to generate a global model update, which is then sent back to the UEs for their subsequent local computation. The FL process terminates when a prescribed level of learning accuracy is attained; at which point, a learning model is established. Here, since only model updates (instead of raw training data) is shared between UEs and a central server, data privacy of each UE is protected. Furthermore, as a model update is much smaller in size than raw training data, sending a model update requires a much shorter amount of time.} 

In the literature, there are two main research directions that study FL in wireless network environments. The learning-oriented direction aims to develop FL frameworks that improve learning performance such as test accuracy. They do so by mitigating the detrimental effects of wireless transmissions, such as channel fading and estimation errors, on FL \cite{chen21TWC,amiri20TWC,chen20ICC,dinhIoT21}. 
On the other hand, the communication-oriented direction aims to develop communication schemes that optimize certain performance metrics for communications. Examples  of  these metrics  include  execution  time  (in  seconds)  and  energy  consumption (in Joule) of an FL process executed ``over-the-air'' \cite{yang21TWC, zeng21TWC, zeng20ICC, hu20WCSP, vu20TWC}.

\begin{figure*}[t!]
\centering
\includegraphics[width=0.55\textwidth]{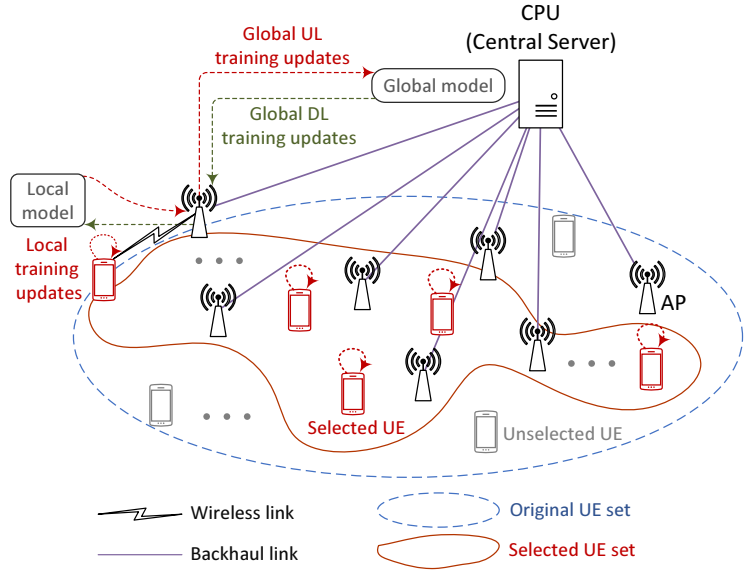}
\vspace{-4mm}
\caption{A CFmMIMO network model for wireless federated learning.}
\label{fig:systemmodel}
\vspace{-2mm}
\end{figure*}

The focus of this paper is on the communication-oriented direction. Here, \cite{yang21TWC, zeng21TWC, zeng20ICC, hu20WCSP} develop  new  wireless  network  designs that use frequency-division multiple access (FDMA) and time-division multiple access (TDMA) to support FL. In these works, the transmission time of each FL communication round could  be  significantly  prolonged  when  the  number  of  UEs is  large. To enable stable and fast FL over wireless media, \cite{vu20TWC,vu21ICC,vu21SPAWC} propose to use cell-free massive multiple-input multiple-output (CFmMIMO) to assist the process of FL.  In a CFmMIMO network, UEs are simultaneously served by a large number of distributed access points (APs) over the same frequency band, and hence, CFmMIMO can offer very high macro diversity and multiplexing gains \cite{ngo17TWC}. As a result, it can uniformly provide very high data rates for all users in the network which enables the stable and fast FL. Paper\cite{vu20TWC} shows that the CFmMIMO network can achieve a much shorter FL execution time than that of conventional TDMA/FDMA networks. 
We note that in a wireless network, the information transmitted is not necessarily private. However, in a wireless network with FL, user privacy is better protected since only the model updates (instead of raw data) are sent from devices to a central processing unit. While the risk of privacy leakage exists in a wireless FL network, it does not overtake the benefits of FL. We also note that more advanced FL algorithms to secure FL in wireless networks have been recently studied (see \cite{liu20WC} and references therein).

\textbf{Scope of Research}: In this paper, we employ CFmMIMO as the underlying wireless network structure and propose novel resource allocation schemes to support a wireless FL process. Here, we use a conjugate beamforming/matched filtering scheme at the APs. Implemented locally at each AP, this scheme has been shown to perform well when the number of APs is large. The overarching objective is to minimize the total FL execution time, which is a product of the number of FL communication rounds and the execution time of each round.

\textbf{UE Selection}: To achieve the above design objective, we propose allowing only a subset of UEs with favorable links to participate in the learning process. Doing so would reduce the execution time of one FL communication round, albeit at the expense of requiring more FL communication rounds for the FL process to converge. Our aim is to devise a UE selection policy that strikes an optimal balance between these two conflicting outcomes, thereby offering a minimum total FL execution time. This aim is different from that of improving the learning performance (such as robustness, test accuracy, and convergence rate) via UE selection, as advocated by \cite{li20ICLR,farzin19arXiv,xia20TWC,mcmahan17AISTATS,sai20ICML}.

It is worth noting that as the UE selection reduces the number of participating UEs, the test accuracy of an FL process is also affected \cite{yang20TWC}. On the other hand, without UE selection, the FL execution time could potentially be prolonged by the UEs with unfavorable link conditions. These UEs need much more time to transmit. Over a large service area and with a large number of UEs, it is more common to have these vulnerable UEs. Of our particular interest is how to find an optimal set of ``sufficiently good" UEs to reduce the FL execution time, without satisfying the test accuracy too much. This important observation leads to two research questions: 
\begin{itemize}
    \item[Q1)] What is the minimum number of participating UEs $N_{\text{QoL}}$, above which the test accuracy remains acceptable? 
    \item[Q2)] For a given $N_{\text{QoL}}$, how to select UEs with favorable links to minimize the total FL execution time?
\end{itemize}
In this paper, we focus on Q2) and leave Q1) for future research.

\textbf{New Optimization-Based Design and Results}: UE selection has tight relationships with the allocation of other resources. Specifically, the decision which UE is selected affects how to allocate transmit power and computing frequency, and vice versa. Therefore, it is not clear and difficult to see how to find manually or heuristically the favourable links and determine the exact number of the selected UEs for minimizing the FL execution time. Motivated by this observation, we propose a joint optimization approach for UE selection and resource allocation.  In particular, to minimize the execution time of one FL process, the FL execution time is first formulated as a function of UE selection, transmit power, and computation frequency. 
A mixed-integer two-stage stochastic nonconvex problem of minimizing the FL execution time is then formulated and subjected to the constraints that reflex the practical relationships among all the variables, and the practical constraints on 
maximum powers at the APs and UEs, imperfect channel estimation, and a minimum number of selected UEs.
The considered problem is different from that in \cite{vu21ICC} and \cite{vu21SPAWC} which aim to reduce the execution time in an FL communication round, rather than minimizing the total FL execution time.

\begin{table*}[t!]
\renewcommand{\arraystretch}{1.0}
\caption{Frequently Used Symbols}
\vspace{-2mm}
\label{table:notation}
\begin{center}
\begin{tabular}{|c|c|}
\hline
\textbf{Symbols}  & \textbf{Definition}
\\
\hline
\hline
$\NN, N$      & Set of UEs and its cardinality\\
\hline
$\MM, M$      & Set of APs and its cardinality\\
\hline
$N_{\text{QoL}}$      & Minimum number of UEs for quality of learning \\
\hline
$T_c$      & (Small-scale) coherence time\\
\hline
$\widetilde{T}_c$        & Large-scale coherence time \\
\hline
$a_k$    & Selection variable of UE $k$ \\
\hline
$\eta_{mk}$    & Downlink power control coefficient of UE $k$ at AP $m$ \\
\hline
$f_k$    & Processing frequency of UE $k$ \\
\hline
$\zeta_k$    & Uplink power control coefficient of UE $k$ \\
\hline
$S_d, S_u$    & Sizes of global and local updates \\
\hline
\end{tabular}
\end{center}
\vspace{-2mm}
\end{table*}

To solve the formulated problem, we propose a new algorithm that is proven to converge to the neighbourhood of its stationary points. The mathematical challenges lie in the binary nature of UE selection variables and the coupling among the optimizing variables. Structure-wise, this optimization problem is much more complex than those considered in \cite{vu20TWC} and \cite{vu21SPAWC}. The algorithm developed in this paper is able to handle the binary constraints efficiently, whilst meeting all the strict conditions for solving two-stage stochastic nonconvex problems \cite{liu18TSP}. Numerical results with practical parameter settings verify the convergence of the proposed algorithm. They also show that our UE selection approach can cut the total FL execution time by more than half compared to the baseline schemes.

{\textbf{Notation and Paper Organization}:} 
In this paper, boldfaced symbols are used for vectors and capitalized boldfaced symbols for matrices.
$\RRR^d$ is a space where its elements are real vectors with length $d$.
$\pmb{X}^*$ and $\pmb{X}^H$ are the conjugate and conjugate transposition of a matrix $\pmb{X}$, respectively.
$\CN(\pmb{0},\pmb{Q})$ denotes the circularly symmetric complex Gaussian distribution with zero mean and covariance $\pmb{Q}$. $\NN(0,V)$ denotes the normal distribution with zero mean and variance $V$. 
$\nabla g$ is the gradient of a function $g$.
$\EEE\{x\}$ denotes the expected value of a random variable $x$.
For ease of reference, the symbols frequently used throughout the paper are listed in Table~\ref{table:notation}.


The rest of this paper is organized as follows. Section~\ref{sec:FL:SystemModel} introduces a CFmMIMO system model to support a standard FL framework. Section~\ref{sec:PF} formulates the optimization problem of minimizing the FL execution time, whereas Section~\ref{sec:alg} proposes a new algorithm to solve this problem.
Section~\ref{sec:sim} verifies the performance of the developed algorithm through numerical examples. Finally, Section~\ref{sec:con} concludes the paper.

\vspace{-3mm}
\section{Cell-Free Massive MIMO System Model to Support Wireless Federated Learning}
\label{sec:FL:SystemModel}
\vspace{-1mm}

We consider the CFmMIMO network model illustrated in Fig.~\ref{fig:systemmodel}. 
Let ${\NN}\triangleq\{1,\dots,N\}$ be the UE set in the network. The UEs are served by a set of APs $\MM=\{1,...,M\}$ via wireless access links with the same time-frequency resource \cite{ngo16,emil17}. The transmission between the APs and UEs are under the time-division duplexing operation with channel reciprocity\footnote{In general, channel reciprocity error can be a performance limiting factor in wireless networks using the time-division duplexing operation. It is because the hardware chains in the transmitter and receivers may not be reciprocal between the uplink and downlink. However, according to \cite{erik14CM}, the calibration of hardware chains for channel reciprocity is not a serious problem in \emph{massive MIMO systems}. For example, there are calibration solutions that have been tested successfully for practical 64-antenna systems \cite{shepard12,florian10}.}.
The APs and UEs are each equipped with a single antenna. These APs are assumed to connect to a central processing unit (CPU) via backhaul links with high-capacity. As such, the transmission times between the CPU and all the APs are considered negligible \cite{ngo17TWC}.
The CPU and UEs act as the central server and the clients in an FL process, respectively. Here, the APs relay the model updates between the CPU and the UEs.

\vspace{-2mm}
\subsection{{UE Selection Model and Standard Federated Learning Framework}}
\label{model:UESelection}

{
First, we propose to select a subset of UEs out of $N$ original UEs to participate in an FL process. 
Let variable $a_k$ indicate whether a UE $k\in\NN\triangleq\{1,\dots,N\}$ is selected to partake in an FL process, i.e.,
\begin{align}\label{a}
a_k \triangleq
\begin{cases}
  1,& \text{if UE $k$ is selected,}\\
  0, & \mbox{otherwise}.
\end{cases}
\end{align}
Let $\widetilde{\NN}$ be the set of selected UEs. Denote by $\widetilde{N}\triangleq\sum_{k\in\NN}a_k\leq N$ the cardinality of $\widetilde{\NN}$ or the number of selected UEs, and
by $\aaa\triangleq[a_1,\dots,a_N]^T$ the vector of UE selection.
{We insist that}
\begin{align}
\label{tildeN}
& \sum_{k\in\NN}a_k \geq N_{\text{QoL}},
\end{align}
where $N_{\text{QoL}}$ is a predefined minimum number of participating UEs, chosen such that an acceptable \emph{quality of learning (QoL)} is achieved. As noted in Sec.~\ref{sec:Introd}, how to determine the exact value of $N_{\text{QoL}}$ is out of the scope of this paper.}

Next, the above $\widetilde{N}$ selected UEs will participate in an FL process that adopts the standard FL framework with a synchronous aggregation mode\footnote{FL algorithms with the synchronous aggregation mode wait to receive all local model updates sent from UEs before aggregation, while the FL algorithms with an asynchronous aggregation mode do not. The FL algorithms with synchronous aggregation normally offer better convergence rate and accuracy than the FL algorithms operating with asynchronous aggregation. FL algorithms with synchronous aggregation are well studied \cite{xu21}, while research on improvement of learning performance of the FL algorithms with asynchronous aggregation is still in its infancy. As such, our paper focuses on resource allocation for supporting FL with synchronous aggregation \cite{canh21TN, sai20ICML, canh20NIPS, felix20TNNLS, tran19INFOCOM, mcmahan17AISTATS}.} \cite{canh21TN, sai20ICML, canh20NIPS, felix20TNNLS, tran19INFOCOM, mcmahan17AISTATS}. In general, it is an iterative process consisting of multiple communication rounds. Each communication round involves the following basic steps (S1)-(S4):
\vspace{-0mm}
\begin{enumerate}[label={(S\arabic*)}]
\item 
A central server sends the global update to all the selected UEs.
\item 
Upon receiving the global update from the central server, the UEs solve their local learning problems (based on their local data set), and compute their local model updates.
\item 
The UEs send their local updates to the central server.
\item 
Upon receiving all the local updates from UEs, the central server computes the global update by aggregation.
\end{enumerate}
The communication rounds repeat until a certain level of test accuracy is attained at the server. At which point, the server terminates the whole FL process.
\vspace{-0mm}



%

\subsection{{CFmMIMO System Model for Wireless Federated Learning}}\label{subsec:CFmMIMO}

{
Now, we will describe in detail how a communication round outlined in Sec.~\ref{model:UESelection} is realized by CFmMIMO.}
 
\textbf{Assumption (A1)}: All selected UEs execute steps (S2) and (S3) in a synchronized manner. Within a step (S2) or (S3), if a UE completes its operation before other UEs, it must wait until all other UEs complete their respective operation before proceeding to the next step.
\vspace{-0mm}

\begin{remark}
It should be noted that synchronising the operations of UEs at the Steps (S2) and (S3) would cause some waste of radio resources because of the wait time among the UEs. Instead, an asynchronous design may potentially reduce the execution time by removing the wait time. However, in the asynchronous design, some UEs may finish Step (S2) while other UEs have not even completed Step (S1) yet. Therefore, it is difficult for the APs to know exactly when they need to switch to the uplink mode to receive local model updates from the UEs. The solutions for overcoming this difficulty
may make the designs of signal processing and signalling much more complex. Based on this observation, this work focuses on the proposed synchronous communication design for ease of implementation, and leaves the asynchronous design for future work.
\end{remark}

{
\textbf{Assumption (A2)}: Each FL communication round is completely executed within a large-scale coherence time of the wireless channel \cite{vu20TWC}.\footnote{A large-scale coherence time is the period of time where the large-scale fading coefficients are reasonably invariant \cite{vu20TWC}.}}

\begin{figure*}[t!]
	\centering
	\includegraphics[width=0.6\textwidth]{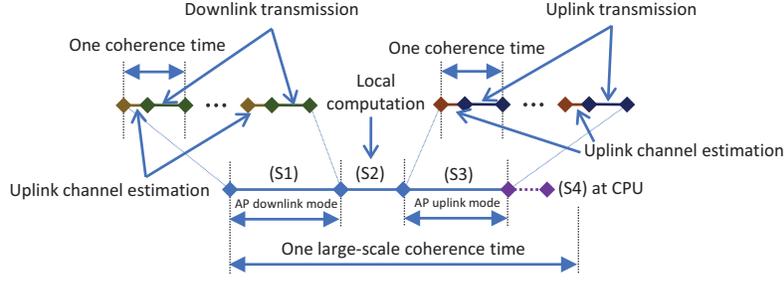}
	\vspace{-0mm}
	\caption{Details of operation of one FL communication round in the considered CFmMIMO network.}
	\label{fig:time2}
\end{figure*}

\begin{remark}
Assumption 2 is realistic in many practical scenarios, such as outdoor systems with  moderately low user mobility (e.g., dense urban areas) and indoor systems (e.g., office buildings, stadiums, cinema theaters, factory). 
An example is a system that supports users' speed of $v \leq 7.5$ m/s $= 27$ km/h. With a conventional carrier frequency $f_c=1$ GHz  \cite{emil17}, the coherence time of the channels is $T_c \geq \frac{c_\ell}{4f_c v} = 10$ ms, where $c_\ell=3\times 10^8$ m/s is the speed of light. According to the measurements in \cite{viering02}, the large-scale coefficients can stay relatively unchanged for at least $100$ times longer than the coherence time. Therefore, the large-scale coherence time is $\widetilde{T}_c \geq 100T_c \geq 1$ s. Let $S = 4$ KB = $32$ Kb be the size of global and local updates \cite{tran19INFOCOM}. For a 5G system with a user-experienced downlink rate of $R_d=100$ Mbps and uplink rate of $R_u = 50$ Mbps \cite{ericsson20WP}, the time durations for downlink and uplink transmissions are $T_d = \frac{S}{R_d} = 0.32$ ms, $T_u = \frac{S}{R_u} = 0.64$ ms, respectively.
Now, let $D_s = 10^6$ be number of data samples at each UE, $c=20$ cycles/sample be the number of CPU cycles required to process one data sample, $f=3\times 10^9$ cycles/s be the frequency of processing CPU cycles, and $N_L=5$ be the number of iterations of processing the local data at UEs \cite{tran19INFOCOM}. 
Then, the time duration needed to compute the local update is $T_c = N_L \frac{cD_s}{f} = 33.3$ ms. Therefore, the execution time of one FL communication round is $T_o = T_d + T_c + T_u = 34.26$ ms $<< 1$ s $\leq \widetilde{T}_c$. Note that the number of collected data samples at UEs in FL applications is normally small (e.g., $10^3-10^4$) \cite{prayitno21}. Therefore, the time of computing local updates should be much shorter than that in our example. This means in practice, the executive time of one FL communication round can even be much smaller than the large-scale coherence time.

On the other hand, Assumption 2 does not impose any practical challenge for deploying cell-free massive MIMO in indoor scenarios. The issues of deployment complexity and high front-haul capacity of cell-free massive MIMO in indoor scenarios can be effectively resolved by using an appropriate architecture, namely \textit{radio stripe system} \cite{giovanni19,shaik21TWC,miretti21SPAWC}.
In this system, each antenna element is effectively an AP. Then, a large number of small-size antennas are put in a cable or a radio stripe. Inside the radio stripe, antennas are connected to their associated antenna processing units (APUs). The APUs and antennas are power-supplied via a shared bus that is serially located in the radio stripe. Finally, each radio stripe is connected to one or multiple central processing units. The radio stripe system turns the long front-haul cables into plug-and-play radio stripes, which effectively and flexibly implement
the star topology and improve the coverage of cell-free massive MIMO. The required front-haul capacity of each radio stripe is proportional to the sum rate of data streams for transmitting to or receiving from the users at the maximum network load. Since there is a large number of antennas in the radio stripe system, the required front-haul capacity can be reduced by reducing the number of users served per radio stripe with a proper design of UE-radio stripe association \cite{giovanni19}.
\end{remark}

{
With the above assumptions, we are ready to implement Steps (S1)-(S4) by CFmMIMO, as illustrated in Fig.~\ref{fig:time2}.}

\subsubsection{Step (S1)}
This step has the following two phases, both executed within a coherence block.

\textbf{Uplink channel estimation:}
Uplink pilot sequences are sent by all the UEs to all the APs simultaneously in each small-scale coherence time to help estimate channels. These estimates will later be used to construct downlink signal beams. Here, we assume time-division duplexing where both uplink and downlink channel information is acquired at the APs via uplink channel estimation \cite{ngo17TWC}.
Denote by $\tau_c$ the number of samples of each coherence block, and by $\tau_{t}$ (samples) the length of one pilot sequence.
Let $\sqrt{\tau_{t}}\VARPHI_{k}\in\C^{\tau_{t}\times 1}$ be the pilot sequence transmitted from a UE $k\in\NN$, where $\|\VARPHI_{k}\|^2\!\!=\!\!1,\forall k\in\NN$.
Denote by $g_{mk} \!=\! (\beta_{mk})^{1/2}\tilde{g}_{mk}$ the channel from a UE $k$ to an AP $m$, where $\beta_{mk}$ and $\tilde{g}_{mk} \sim \CN(0,1)$ are the large-scale fading and small-scale fading channel coefficients, respectively. At the AP $m$, $g_{mk}$ is estimated by using the received pilots and  the minimum mean-square error (MMSE) estimation. The MMSE estimate $\hat{g}_{mk}$ of $g_{mk}$ is a random variable distributed according to $\CN(0,\hat{\sigma}_{mk}^2)$, where
$\hat{\sigma}_{mk}^2=\frac{\tau_{t}\rho_{t}(\beta_{mk})^2}
{\sum_{\ell\in\NN}\tau_{t}\rho_{t}\beta_{m\ell}|\VARPHI_k^H\VARPHI_\ell|^2+1}$ \cite{ngo17TWC}. 

\textbf{Downlink transmission}:
Denote by $S_d$ (bits) the data size of the (same) global downlink model update for each selected UE. Since the global downlink model update can include thousands of model weights, $S_d$ is typically large. At the CPU, the global model update intended for UE $k$ is encoded into many symbols $s_{d,k,i} \sim \CN(0,1), \forall i\in\{1,\dots,L_k\}$. The number of symbols $L_k$ depends on the data size of the model update as well as the data rate. For ease of presentation, hereafter, we drop the index $i$ from $s_{d,k,i}$ and let $R_{d,k}$ be the data rate of the global update intended for a UE $k\in\widetilde{\NN}$.

To transmit the symbols received from the CPU, the APs first use conjugate beamforming to precode these symbols before broadcasting them to all the selected UEs. Specifically, the transmitted signal at an AP $m$ is given as $x_{d,m}\!\!=\!\! \sqrt{\rho_{d}}\sum_{k\in\NN}\sqrt{\eta_{mk}}(\hat{g}_{mk})^*s_{d,k}$,
where $\rho_{d}$ is the maximum normalized transmit power at each AP
and $\eta_{mk}, \forall m\in\MM,k\in\NN$, is a power control coefficient.
The transmitted power at the AP $m$ is required to meet the average normalized power constraint, i.e., $\EEE\{|x_{d,m}|^2\}\leq \rho_d$, which can also be expressed as the following per-AP power constraint:
\begin{align}\label{power:d:cons}
\sum_{k\in\NN}\sigma_{mk}^2\eta_{mk} \leq 1, \forall m.
\end{align}
Since no power should be allocated to the non-selected UEs, we have
\begin{align}\label{a:DL}
\forall k\in\NN: \text{if}\,\,a_k = 0,\,\,\text{then}\,\,\forall m\in\MM, \eta_{mk} = 0.
\end{align}

\begin{figure*}[t!]
\begin{align}\label{R:d}
 R_{d,k}(\ETA)\!=\!\frac{\tau_c\!-\!\tau_t}{\tau_c}B\log_2\!\Bigg(1\!+\!\frac
{\rho_d\big(\sum_{m\in\MM}\eta_{mk}^{1/2}\hat{\sigma}_{mk}^2\big)^2}
{
\underbrace{
\rho_d\sum_{\ell\in\NN\setminus k}
\big(\sum_{m\in\MM}\eta_{m\ell}^{1/2}\hat{\sigma}_{m\ell}^2\frac{\beta_{mk}}{\beta_{m\ell}}\big)^2
|\VARPHI_\ell^H\VARPHI_k|^2
}_{\text{Pilot contamination}
}
\!+\!\underbrace{\rho_d\sum_{\ell\in\NN}\sum_{m\in\MM}\eta_{m\ell}\hat{\sigma}_{m\ell}^2\beta_{mk}}_{\text{Inter-user interference}}
+1}\Bigg)
\end{align}
\vspace{-0mm}
\end{figure*}

The achievable downlink data rate at a UE $k$ is given in \eqref{R:d} \cite[(24)]{ngo17TWC} (shown at the top of next page),
where $\ETA\triangleq\{\eta_{mk}\}_{m\in\MM,k\in\NN}$, and $B$ is the bandwidth.
The numerator in the log function represents the power of the desired signals normalized by noise. The three terms in the denominator in the log function represent the normalized-by-noise power of the pilot contamination, inter-UE interference and noise, respectively.

The transmission time from the APs to a UE $k$ is given by
\vspace{-0mm}
\begin{align}\label{latency:d:wireless}
t_{d,k}(a_k,\ETA) = \frac{a_kS_d}{R_{d,k}(\ETA)},
\end{align}
where $a_k\neq 0$. Here, \eqref{latency:d:wireless} captures the fact that the non-selected UEs do not receive any intended downlink transmission, i.e., $t_{d,k} = 0$, for any $k\in\NN$ with $a_k = 0$.

\subsubsection{Step (S2)}
\label{S2}
After receiving the global update, each UE $k$ computes its local model update by using its local dataset.

\textbf{Computation delay}: Let $c_k$ (cycles/sample) be the number of processing cycles required to process one data sample at a UE $k$. Assume that
$c_k$ is known \emph{a priori} by an offline measurement \cite{miettinen10}.
Denote by $D_k$ (samples) and $f_k$ (cycles/s) the size of the local data set and the processing frequency of the UE $k$, respectively.
The total computing time at the UE $k$ is expressed as \cite{tran19INFOCOM}
\begin{align}
    t_{c,k}=\frac{a_kLD_kc_k}{f_k},
\end{align}
where 
$a_k \neq 0$, 
$L$ is the number of local computing iterations, and $\frac{D_kc_k}{f_k}$ is the computing time of each iteration over the local training data set at the UE $k$. Similarly, because the non-selected UEs do not compute their local models, we have 
\begin{align}
  \forall k\in\NN: a_k = 0, t_{c,k} = 0.
\end{align}

\subsubsection{Step (S3)}
This step has the following two phases, both executed within a coherence block.

\textbf{Uplink channel estimation}: This phase uses the same pilot assignment and channel estimation techniques in the uplink channel estimation of Step (S1). The MMSE estimate $\bar{g}_{mk}$ of $g_{mk}$ is thus a random variable distributed according to $\CN(0,\bar{\sigma}_{mk}^2)$, where
$\bar{\sigma}_{mk}^2=\frac{\tau_{t}\rho_{t}(\beta_{mk})^2}
{\sum_{\ell\in\NN}\tau_{t}\rho_{t}\beta_{m\ell}|\VARPHI_k^H\VARPHI_\ell|^2+1}$.

\textbf{Uplink transmission}: Similar to Step (S1), for ease of presentation, a symbol $s_{u,k} \sim \CN(0,1)$ is encoded for the local model update of a UE $k$.
The symbol $s_{u,k}$ is then allocated a transmit amplitude value $\sqrt{\rho_{u}\zeta_k}$ to generate a baseband signal $x_{u,k}$ for wireless transmission, i.e., $x_{u,k}=\sqrt{\rho_{u}\zeta_k}s_{u,k}$. The UE $k$ is subjected to the average transmit power constraint, i.e., $\EEE\left\{|x_{u,k}|^2\right\}\leq \rho_u$, which can also be expressed as a per-UE constraint:
\begin{align}\label{power:u:cons}
0\leq\zeta_k\leq 1,\forall k\in\NN.
\end{align}
Because no power should be allocated to the non-selected UEs, we have
\begin{align}\label{a:UL}
\forall k\in\NN: \text{if}\,\, a_k = 0,\,\,\text{then}\,\, \zeta_{k} = 0.
\end{align}
The achievable uplink data rate $R_{u,k}$ at the CPU for the UE $k$ is then given in \eqref{R:u} shown at the top of next page, where $\ZETA\triangleq\{\zeta_{k}\}_{k\in\NN}$ \cite[Eq. (27)]{ngo17TWC}.
\begin{figure*}[t!]
\begin{align}\label{R:u}
R_{u,k}(\ZETA)\!\!=\!\!\frac{\tau_c\!-\!\tau_t}{\tau_c}B \log_2 \!\!\Bigg(\!1+\!
\frac
{\rho_u\zeta_{k}\left(\sum_{m\in\MM}\bar{\sigma}_{mk}^2\right)^2}
{\underbrace{\rho_u\!\!\!\sum_{\ell\in\NN\setminus k}\!\!\!\zeta_{\ell}
\big(\!\!\sum_{m\in\MM}\!\! \bar{\sigma}_{mk}^2\frac{\beta_{m\ell}}{ \beta_{mk}}\big)^2
|\VARPHI_k^H\VARPHI_\ell|^2}_{\text{Pilot contamination}}
\!+\!\underbrace{\rho_u\!\sum_{\ell\in\NN}\zeta_{\ell}\!\!\sum_{m\in\MM}\!\!\bar{\sigma}_{mk}^2\beta_{m\ell}}_{\text{Inter-user interference}}
\!+\!\sum_{m\in\MM}\bar{\sigma}_{mk}^2}
\Bigg)
\end{align}
\hrulefill
\end{figure*}

Let $S_u$ (bits) and $R_{u,k}$ (bps) be the same data size of the local model updates and the data rate of transmitting the local model update from a UE $k$ to the CPU, respectively. The transmission time from the UE $k$ to the APs is given by
\begin{align}\label{latency:u:wireless}
t_{u,k}(a_k,\ZETA) = \frac{a_k S_u}{R_{u,k}(\ZETA)},
\end{align}
where $a_k \neq 0$. Since the non-selected UEs do not have any uplink transmission, we have $\forall k\in\NN: a_k = 0, t_{u,k} = 0$.

\subsubsection{Step (S4)} The CPU computes the global update by using all the received local updates. Since the CPU's computational capability is much higher than that of the UEs, the time required to compute the global update is assumed negligible.


\vspace{-0mm}
\section{{Joint UE Selection and Resource Allocation for Minimizing Total FL Execution Time in CFmMIMO: Optimization Problem Formulation}}
\label{sec:PF}

\subsection{{Total Execution Time of an FL Process}}
{
The execution time of one FL communication round comprises the execution times of Steps (S1)--(S3) [as Step (S4) is assumed not to incur any time delay]. The execution time of each of these steps is the longest execution time within that step. As such, the execution time of one FL communication round is expressed as
\begin{align}\label{time:oneiteration}
\nonumber
&T_o(\aaa,\ETA,\f,\ZETA)
\\
\nonumber
&\triangleq \max_{k}t_{d,k}(a_k,\ETA)
+ \max_{k}t_{c,k}(a_k,f_k)
+ \max_{k}t_{u,k}(a_k,\ZETA)
\\
&\triangleq \max_{k}\frac{a_k S_d}{R_{d,k}(\ETA)}
+ \max_{k}\frac{a_kLD_kc_k}{f_k}
+ \max_{k}\frac{a_k S_u}{R_{u,k}(\ZETA)},
\end{align}
where $k\in\{k\in\NN | \,\, a_k \neq 0\}$. Note that $\aaa$ (UE selection) is determined before an FL process is executed, while $\ETA,\f$, and $\ZETA$ are optimized in each FL communication round. The mathematical expression of the execution time of one FL communication round must have a two-timescale structure. 
Therefore, 
we introduce a new metric named ``effective execution time of one FL communication round" $\EEE\{T_o(\aaa,\ETA,\f,\ZETA)\}$, which is the average of $T_o(\aaa,\ETA,\f,\ZETA)$ over large-scale fading realizations. The values of $\aaa$ in $\EEE\{T_o(\aaa,\ETA,\f,\ZETA)\}$ remains unchanged, while those of $\ETA,\f$, and $\ZETA$ in $T_o(\aaa,\ETA,\f,\ZETA)$ are optimized in each large-scale coherence time. 
}

{Let $G(\aaa)$ be the number of communication rounds in an FL process.} By using the tightest convergence rates of FL \cite[Theorem~I]{sai20ICML}, we have
\begin{align}\label{G}
    G(\aaa) = \frac{q}{\sum_{k\in\NN}a_k}.
\end{align}
Here, $q$ is a known constant which depends on the specific characteristics of the FL learning problems. The total execution time of an FL process is then given by
\begin{align}
\label{Tall}
T_e(\aaa,\ETA,\f,\ZETA)\triangleq G(\aaa) \EEE\{T_o(\aaa,\ETA,\f,\ZETA)\}.
\end{align}

{
Eqs.~\eqref{time:oneiteration}--\eqref{Tall} clearly shows the effects of UE selection on the overall execution time. According to \eqref{time:oneiteration}, if we only allow users with favorable link conditions (i.e., strong channel gains, weak pilot contamination and interference) to partake in the FL process, the execution time of one FL communication round is shorter. However, selecting only a subset of all UEs also means increasing the required number of communication rounds as shown in \eqref{G}. Therefore, minimizing the total execution time in \eqref{Tall} involves finding an optimal set of UEs to be selected, in order to balance between the two contradicting effects.}

\vspace{-0mm}
\subsection{Problem Formulation}

We are now ready to formulate the main design problem as the following optimization problem.
\begin{subequations}\label{mainP1}
\begin{align}
\label{CF:mainP1}
\underset{\aaa,\ETA,\f,\ZETA}{\min} \,\,
&T_e(\aaa,\ETA,\f,\ZETA)
\\
\nonumber
\mathrm{s.t.}\,\,
&
\eqref{a}-\eqref{a:DL}, \eqref{power:u:cons},\eqref{a:UL}
\\
\label{cons:eta}
& 0\leq \eta_{mk}, \forall m,k
\\
\label{cons:zeta}
& 0\leq \zeta_{k}, \forall k
\\
\label{fbound}
& 0\leq f_k \leq f_{k,max}, \forall k.
\end{align}
\end{subequations}
Problem \eqref{mainP1} has a nonconvex stochastic, mixed-integer mixed-timescale structure, along with binary constraints and tight coupling among the optimizing variables. Finding its globally optimal solution is challenging.

\vspace{-0mm}
\section{Joint UE Selection and Resource Allocation for Minimizing Total FL Execution Time in CFmMIMO: Proposed Algorithm}
\label{sec:alg}
First, to deal with the binary constraint \eqref{a}, we observe that $x\in\{0,1\}\Leftrightarrow x\in[0,1]\, \text{and} \, x-x^2\leq0$ \cite{vu18TCOM,vu18TWC}. Therefore, \eqref{a} is equivalent to the following two constraints
\begin{align}
\label{suma}
& \sum_{k\in\NN}(a_k-a_k^2)\leq 0
\\
\label{arelax}
&0\leq a_k\leq 1, \forall k.
\end{align}
Since the values of $\{a_k\}$ in \eqref{suma} and \eqref{arelax} are now real, it is easier to handle \eqref{suma} and \eqref{arelax} than \eqref{a}.
We then use \eqref{suma} and \eqref{arelax} to rewrite problem \eqref{mainP1} as
\begin{align}\label{mainP1:equiv}
\underset{\aaa,\ETA,\f,\ZETA}{\min} \,\,
&T_e(\aaa,\ETA,\f,\ZETA)
\\
\nonumber
\mathrm{s.t.}\,\,
&
\eqref{tildeN}-\eqref{a:DL}, \eqref{power:u:cons},\eqref{a:UL},
\eqref{cons:eta}-\eqref{fbound},
\eqref{suma}, \eqref{arelax}.
\end{align}

Next, to deal with the $\max$ functions in $T_o(\aaa,\ETA,\f,\ZETA)$, we rewrite problem \eqref{mainP1:equiv} in a more tractable epigraph form as
\begin{subequations}\label{mainP1:epi}
\begin{align}
\underset{\x}{\min} \,\,
&G(\aaa)\EEE\{\widetilde{T}_o(t_d,t_c,t_u)
\}
\\
\nonumber
\mathrm{s.t.}\,\,
&
\eqref{tildeN}-\eqref{a:DL}, \eqref{power:u:cons},\eqref{a:UL},
\eqref{cons:eta}-\eqref{fbound},
\eqref{suma}, \eqref{arelax}
\\
\label{rlowerbound}
&0 \leq r_{d,k}, 0 \leq r_{u,k}, \forall k
\\
\label{rdupperbound}
&r_{d,k}\leq R_{d,k}(\ETA), \forall k
\\
\label{ruupperbound}
&r_{u,k}\leq R_{u,k}(\ZETA), \forall k
\\
\label{td}
&\frac{a_k S_d}{r_{d,k}} \leq t_d, \forall k, a_k \neq 0
\\
\label{tc}
&\frac{a_k L D_kc_k}{f_k} \leq t_c, \forall k, a_k \neq 0
\\
\label{tu}
& \frac{a_kS_u}{r_{u,k}} \leq t_u, \forall k, a_k \neq 0,
\end{align}
\end{subequations}
where $\widetilde{T}_o(t_d,t_c,t_u)=t_d + t_c +t_u$, $\x\triangleq\{\aaa,\ETA,\f,\ZETA,\rr_d,\rr_u,t_d,\\ t_c,t_u\}$;
$\rr_d \triangleq \{r_{d,k}\}$, $\rr_u \triangleq \{r_{u,k}\}, \forall k \in \NN$,
$ t_d, t_c$, and $t_u$ are additional variables. 

Our main problem \eqref{mainP1} is now transformed to problem \eqref{mainP1:epi} which is a stochastic nonconvex optimization problem. In \eqref{mainP1:epi}, the UE selection variables $\aaa$ are optimized in a long-term timescale (before any FL process happens), while $\widetilde{\x}\triangleq\x\setminus \aaa$ including transmit power and processing frequency are optimized in a short-term timescale (in each FL communication round). To solve this type of optimization problem, we adopt the general framework of \cite{liu18TSP}. Specifically, we decompose problem \eqref{mainP1:epi} into a short-term subproblem and a long-term master problem, and solve these resulting problems in an alternating manner. The mathematical derivations are detailed in the following.

For a given $\aaa$, in each large-scale coherence time, the \emph{short-term subproblem} is expressed as:
\begin{align}\label{shortP}
\underset{\widetilde{\x}}{\min} \,\,
&\widetilde{T}_o(t_d,t_c,t_u)
\\
\nonumber
\mathrm{s.t.}\,\,
&
\eqref{power:d:cons},\eqref{a:DL}, \eqref{power:u:cons},\eqref{a:UL},
\eqref{cons:eta}-\eqref{fbound},
\eqref{rlowerbound}-\eqref{tu}.
\end{align}
For given optimal solutions $\widetilde{\x}$ to problems \eqref{shortP}, we have $t_d=\frac{a_{k^*}S_d}{ R_{d,k^*}}$, where $k^*\triangleq\underset{k\in\NN}{\mathrm{argmax}}\frac{a_kS_d}{ R_{d,k}}$.
Therefore, we can obtain $t_d = \aaa^T\tilde{\ttt}_d$, where $\tilde{\ttt}_d$ is the vector whose elements are $0$ except for the $k^*$-th element, and the value of this element is $\frac{S_d}{ R_{d,k^*}}$. Similarly, we have $t_c^*=\frac{a_{i^*} L D_{i^*}c_{i^*}}{ f_{i^*}}$ and $t_u = \frac{a_{j^*}S_u}{ R_{u,j^*}}$, where $i^*\triangleq \underset{k\in\NN}{\mathrm{argmax}}\frac{a_k L D_kc_k}{ f_k}$ and $j^* \triangleq \underset{k\in\NN}{\mathrm{argmax}}\frac{a_k S_u}{ R_{u,k}}$. Then, we also have $t_c = \aaa^T\tilde{\ttt}_c$ and $t_u = \aaa^T\tilde{\ttt}_u$. Here,  $\tilde{\ttt}_c$ is the vector whose elements are $0$ except for the $i^*$-th element, and the value of this element is $\frac{L D_{i^*}c_{i^*}}{ f_{i^*}}$. $\tilde{\ttt}_u$ is the vector whose elements are $0$ except for the $j^*$-th element, and the value of this element is $\frac{S_u}{w_{j^*} R_{u,j^*}}$.
Then, the \emph{long-term master problem} is expressed as:
\begin{align}\label{longP}
\underset{\aaa}{\min} \,\,
&g(\aaa)
\\
\nonumber
\mathrm{s.t.}\,\,
&
\eqref{tildeN},\eqref{suma},\eqref{arelax}.
\end{align}
where 
\begin{align}
\nonumber
g(\aaa) &\triangleq T_e(\aaa) = \EEE\{T(\aaa)\}
\\
\nonumber
T(\aaa) &\triangleq G(\aaa)\widetilde{T}_o(\aaa)= \frac{q(\aaa^T\tilde{\ttt}_d
+\aaa^T\tilde{\ttt}_c
+\aaa^T\tilde{\ttt}_u)}{\aaa^T\pmb{1}}
\\
\nonumber
\widetilde{T}_o(\aaa) &=\aaa^T\tilde{\ttt}_d
+\aaa^T\tilde{\ttt}_c
+\aaa^T\tilde{\ttt}_u
\\
\nonumber
G(\aaa)&=\frac{q}{\aaa^T\pmb{1}},
\end{align}
and $\pmb{1}\in\RRR^N$ is an all-one vector.

\subsection{Solving Short-term Subproblem \eqref{shortP}}\label{subsec:STO}
\vspace{-0mm}
First, we rewrite problem \eqref{shortP} as
\vspace{-0mm}
\begin{subequations}\label{shortP:epi}
\begin{align}
\label{CF:shortP:epi}
\underset{\widehat{\x}}{\min} \,\,
&\widetilde{T}_o(t_d,t_c,t_u)
\\
\mathrm{s.t.}\,\,
\nonumber
&
\eqref{cons:eta}-\eqref{fbound},
\eqref{rlowerbound}-\eqref{tu}
\\
\label{power:d:cons:1}
& \sigma_{mk}^2\eta_{mk} \leq \tilde{v}_{mk}, \forall m,k
\\
\label{power:d:cons:2}
& \tilde{v}_{mk}\leq a_k, \forall m,k
\\
\label{power:d:cons:3}
& \sum_{k\in\NN}\tilde{v}_{mk}\leq 1, \forall m
\\
\label{power:u:cons:1}
& \zeta_{k}\leq a_k, \forall k,
\end{align}
\end{subequations}
where $\widehat{\x}\triangleq\{\widetilde{\x},\tilde{\vv}\}$ and $\tilde{\vv}\triangleq\{\tilde{v}_{mk}\}_{m\in\MM,k\in\NN}$ are additional variables. Here, \eqref{power:d:cons:1}--\eqref{power:d:cons:3} follow from \eqref{power:d:cons} and \eqref{a:DL}, while \eqref{power:u:cons:1} follows from \eqref{power:u:cons} and \eqref{a:UL}. Problem \eqref{shortP:epi} is still challenging because of the nonconvex constraints \eqref{rdupperbound} and
\eqref{ruupperbound}. To deal with these constraints, we let
$\vv \triangleq \{v_{mk}\}_{m\in\MM,k\in\NN}$ and $\uu\triangleq \{u_k\}_{k\in\NN}$ with
\vspace{-0mm}
\begin{align}\label{v}
v_{mk}&\triangleq \eta_{mk}^{1/2}, \forall m,k,
\\
\label{u}
u_k &\triangleq \zeta_k^{1/2},\forall k,
\end{align}
and rewrite \eqref{shortP:epi} as
\vspace{-0mm}
\begin{subequations}\label{shortP:equi}
\begin{align}
\label{CF:shortP:equi}
\underset{\bar{\x}}{\min} \,\,
&\widetilde{T}_o(t_d,t_c,t_u)
\\
\nonumber
\mathrm{s.t.}\,\,
&\eqref{fbound},\eqref{rlowerbound},\eqref{td}-\eqref{tu},
\eqref{power:d:cons:2}, \eqref{power:d:cons:3}
\\
\label{power:d:cons:1:new}
& \sigma_{mk}^2v_{mk}^2 \leq\tilde{v}_{mk}, \forall m,k
\\
\label{cons:v}
& 0\leq v_{mk},\forall m,k
\\
\label{power:u:cons:1:new}
& u_{k}^2\leq a_k, \forall m,k
\\
\label{cons:u}
& 0\leq u_k\leq 1,\forall k
\\
\label{cons:Rd:new}
& 0\leq r_{d,k} \leq R_{d,k}(\vv), \forall k
\\
\label{cons:Ru:new}
& 0\leq r_{u,k} \leq R_{u,k}(\uu), \forall k.
\end{align}
\end{subequations}
where $\bar{\x}\triangleq\{\widehat{\x},\vv,\uu\}\setminus\{\ETA,\ZETA\}$. Here, \eqref{power:d:cons:1:new}-\eqref{cons:u} follow \eqref{cons:eta}, \eqref{cons:zeta}, \eqref{power:d:cons:1}, \eqref{power:u:cons:1}, \eqref{v}, \eqref{u}, whereas
\eqref{cons:Rd:new} and \eqref{cons:Ru:new} follow \eqref{rdupperbound} and
\eqref{ruupperbound}.

Regarding the nonconvex constraints \eqref{cons:Rd:new} and \eqref{cons:Ru:new}, the concave lower bound $\widetilde{R}_{d,k}(\vv)$ of $R_{d,k}(\vv)$ is given by \eqref{hd:apprx} \cite{nguyen17TCOM} (see the top of the next page),
\begin{figure*}
\begin{align}\label{hd:apprx}
\widetilde{R}_{d,k}(\vv) \!\triangleq\! \frac{\tau_c\!-\!\tau_t}{\tau_c\log 2}B \Bigg[
\log\left(1\!+\!\frac{(\Upsilon_k^{(\kappa)})^2}{\Pi_k^{(\kappa)}}\right)
\!-\!\frac{(\Upsilon_k^{(\kappa)})^2}{\Pi_k^{(\kappa)}}
\!+\!2\frac{\Upsilon_k^{(\kappa)}\Upsilon_k}{\Pi_k^{(\kappa)}}
\!-\! \frac{(\Upsilon_k^{(\kappa)})^2(\Upsilon_k^2+\Pi_k)}{\Pi_k^{(\kappa)}((\Upsilon_k^{(\kappa)})^2+\Pi_k^{(\kappa)})}\Bigg]
\!\leq\! R_{d,k}(\vv),
\end{align}
\end{figure*}
where 
\begin{align}
\nonumber
&\Upsilon_k(\{v_{mk}\}_{m\in\MM}) = \sqrt{\rho_d}
\sum_{m\in\MM}v_{mk}\sigma_{mk}^2 
\\
\nonumber
& \Pi_k(\vv) =
\rho_d\sum_{\ell\in\NN\setminus k}
\big(\sum_{m\in\MM}v_{m\ell}\sigma_{m\ell}^2\frac{\beta_{mk}}{\beta_{m\ell}}\big)^2
|\VARPHI_\ell^H\VARPHI_k|^2
\\
\nonumber
& \qquad \qquad +\rho_d\sum_{\ell\in\NN}\sum_{m\in\MM}v_{m\ell}^2\sigma_{m\ell}^2\beta_{mk}+1.
\end{align}
Similarly, the concave lower bound $\widetilde{R}_{u,k}(\uu)$ of $R_{u,k}(\uu)$ is given by \eqref{hu:apprx} at the top of the next page \cite{nguyen17TCOM}, where 
\vspace{-0mm}
\begin{figure*}
\begin{align}\label{hu:apprx}
\widetilde{R}_{u,k}(\uu)\! \triangleq\!
& \frac{\tau_c\!-\!\tau_t}{\tau_c\log 2}B \Bigg[ \log\left(1\!+\!\frac{(\Psi_k^{(\kappa)})^2}{\Xi_k^{(\kappa)}}\right)
\!-\!\frac{(\Psi_k^{(\kappa)})^2}{\Xi_k^{(\kappa)}}+
2\frac{\Psi_k^{(\kappa)}\Psi_k}{\Xi_k^{(\kappa)}}
\!-\! \frac{(\Psi_k^{(\kappa)})^2(\Psi_k^2+\Xi_k)}{\Xi_k^{(\kappa)}((\Psi_k^{(\kappa)})^2+\Xi_k^{(\kappa)})} \Bigg]
\!\leq\! R_{u,k}(\uu),
\end{align}
\hrulefill
\end{figure*}
\begin{align}
\nonumber
& \Psi_k(u_k) =
\sqrt{\rho_u} u_k(\sum_{m\in\MM}\sigma_{mk}^2)
\\
\nonumber
& \Xi_k(\uu) =
\rho_u\sum_{\ell\in\NN\setminus k}u_{\ell}^2
\big(\sum_{m\in\MM}\sigma_{mk}^2\frac{\beta_{m\ell}}{\beta_{mk}}\big)^2
|\VARPHI_k^H\VARPHI_\ell|^2
\\
\nonumber
& \qquad\qquad +\rho_u\sum_{\ell\in\NN}u_{\ell}^2\sum_{m\in\MM}\sigma_{mk}^2\beta_{m\ell}
+\sum_{m\in\MM}\sigma_{mk}^2.
\end{align}
As such, \eqref{cons:Rd:new} and \eqref{cons:Ru:new} can be approximated by the following convex constraints
\vspace{-0mm}
\begin{align}
\label{cons:Rd:new:appr}
r_{d,k}&\leq \widetilde{R}_{d,k}(\vv), \forall k\in\NN
\\
\label{cons:Ru:new:appr}
r_{u,k}&\leq \widetilde{R}_{u,k}(\uu), \forall k\in\NN.
\end{align}

At iteration $(\kappa+1)$, for a given point $\bar{\x}^{(\kappa)}$, problem \eqref{shortP:equi} (hence \eqref{shortP}) can finally be approximated by the following convex problem:
\vspace{-0mm}
\begin{align}\label{mainP:appr}
\underset{\bar{\x}\in\widetilde{\FF}}{\min} \,\,
\widetilde{T}_o(t_d,t_c,t_u),
\end{align}
where $\widetilde{\FF}\!\triangleq\!\{\eqref{fbound},
\eqref{rlowerbound},\eqref{td}-\eqref{tu},
\eqref{power:d:cons:2}, \eqref{power:d:cons:3},
\eqref{cons:v}-\eqref{cons:u},\eqref{cons:Rd:new:appr}, \eqref{cons:Ru:new:appr}
\}$ is a convex feasible set.
In Algorithm~\ref{alg:2}, we outline the main steps to solve problem \eqref{shortP}.
Let $\FF\triangleq\{\eqref{fbound},
\eqref{rlowerbound},\eqref{td}-\eqref{tu},
\eqref{power:d:cons:2},\eqref{power:d:cons:3},
\eqref{power:d:cons:1:new}-\eqref{cons:Ru:new}\}$ be the feasible set of \eqref{shortP:equi}.
Starting from a random point $\bar{\x}\in\FF$, we solve \eqref{mainP:appr} using CVX \cite{cvx} to obtain its optimal solution $\bar{\x}^*$. This solution is then used as an initial point in the next iteration. The algorithm terminates when an accuracy level of $\varepsilon$ is reached.
\begin{proposition}
\label{shortP:Prop}
Algorithm~\ref{alg:2} converges to a Karush-Kuhn-Tucker (KKT) solution of \eqref{shortP}.
\end{proposition}
\begin{proof}
It is true that $\widetilde{R}_{d,k}(\vv)$ and $\widetilde{R}_{u,k}(\uu)$ satisfy the key properties of general inner approximation functions \cite[Properties (i), (ii), and (iii)]{Marks78OR}. The feasible set $\widetilde{\FF}$ also satisfies the Slater's constraint qualification condition for convex programs.
Therefore, Algorithm~\ref{alg:2} converges to a KKT solution of \eqref{shortP:equi} when starting from a point $\widetilde{\x}^{(0)}\in\FF$ \cite[Theorem 1]{Marks78OR}. By using the variable transformations \eqref{v} and \eqref{u}, it can be seen that the KKT solutions of \eqref{shortP:equi} satisfy the KKT conditions of \eqref{shortP:epi} as well as of \eqref{shortP}.
\end{proof}

\begin{algorithm}[!t]
\caption{Successive convex approximation approach for solving the short-term subproblem \eqref{shortP}}
\begin{algorithmic}[1]\label{alg:2}
\STATE \textbf{Initialization}: Set $\kappa=0$ and choose a random point $\bar{\x}^{(0)}\in\FF$.
\REPEAT
\STATE Update $\kappa=\kappa+1$
Solve the approximated problem \eqref{mainP:appr} of \eqref{shortP} using CVX to obtain its optimal solution $\bar{\x}^*$ and update $\bar{\x}^{(\kappa)}=\bar{\x}^*$
\UNTIL{convergence}
\end{algorithmic}
\end{algorithm}
\vspace{-4mm}

\subsection{Solving the Long-Term Master Problem \eqref{longP}}

Let $V(\aaa)\triangleq \sum_{k\in\NN}(a_k-a_k^2)=\aaa^T(\pmb{1}-\aaa)$, then \eqref{suma} becomes $V(\aaa)\leq 0$. To deal with this non-convex constraint, we consider the problem
\begin{align}\label{longP:appr:relax}
\underset{\aaa}{\min} \,\,
&\LL(\aaa)\triangleq g(\aaa)+\lambda V(\aaa)
\\
\nonumber
\mathrm{s.t.}\,\,
& \eqref{tildeN}, \eqref{arelax},
\end{align}
where $\LL(\aaa)$ is the Lagrangian of \eqref{longP}, $\lambda \geq 0$ is the Lagrangian multiplier  corresponding to \eqref{suma}.
Let $\HHH\triangleq\{\eqref{tildeN},\eqref{arelax}\}$ be the feasible set of  problem \eqref{longP:appr:relax}.
\begin{proposition}
\label{Prop:LongP1}
The following statement holds:
\renewcommand{\labelenumi}{(\roman{enumi})}
\begin{enumerate}
\item {The value $V_\lambda$ of $V$ at the solution of \eqref{longP:appr:relax} corresponding to $\lambda$ is decreasing to $0$ as $\lambda\rightarrow +\infty$}.
  \item Problem \eqref{longP:appr:relax} has the following property, i.e.,
\begin{equation}\label{Strong:Dualitly:hold}
\underset{\aaa\in\HHH}{\min}\,\,
g(\aaa)
=
\underset{\lambda\geq0}{\sup}\,\,
\underset{\aaa\in\widehat{\HHH}}{\min}\,\,
\LL(\aaa,\lambda)
\end{equation}
and is therefore equivalent to \eqref{longP} at the optimal solution $\lambda^*\geq0$ of the sup-min problem in \eqref{Strong:Dualitly:hold}.
\end{enumerate}
\end{proposition}
\begin{proof}
The proof follows from \cite[Proposition 1]{vu18TWC} and \cite[Proposition 1]{vu18TCOM}; hence, omitted for brevity.
\end{proof}

Theoretically, it is required to have $V_\lambda=0$ in order to obtain an optimal $\lambda^*$. According to Proposition~\ref{Prop:LongP1}, $V_\lambda$ decreases to $0$ as $\lambda\to+\infty$. Since there is always a numerical tolerance in computation, it is sufficient to accept $V_\lambda<\varepsilon$ for some small $\varepsilon$ with a sufficiently large value of $\lambda$ chosen.
In our numerical experiment, for $\varepsilon = 10^{-3}$, we see that $\lambda=1$ is enough to ensure $V_\lambda\leq\varepsilon$. Note that this way of choosing $\lambda$ has been widely used in the literature, e.g., \cite{vu18TWC,vu18TCOM,che14TWC,Rashid14TCOM}.

\begin{figure*}[t!]
\centering
\includegraphics[width=0.6\textwidth]{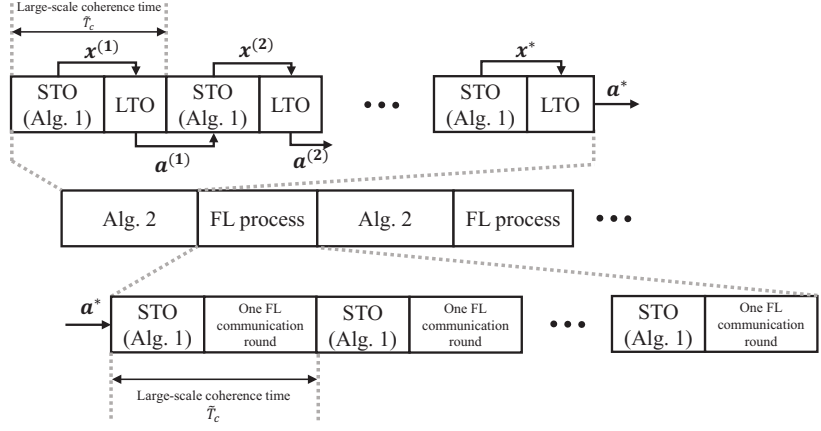}
\vspace{-0mm}
\caption{{Overview of Algorithm~\ref{alg:main}'s operation to assist CFmMIMO-based wireless FL. 
}}
\label{fig:time1}
\end{figure*}

{Problem \eqref{longP:appr:relax} is still challenging due to the expectation operator in $g(\aaa)$ in its cost function $\LL(\aaa)$. 
Following the procedure proposed in \cite{liu18TSP}, we approximate $\LL(\aaa)$ by its surrogate function $\widetilde{\LL}(\aaa)$, which is given as}
{
\begin{align}
\label{LLtilden+1}
\widetilde{\LL}^{(n+1)}(\aaa) \triangleq \,\,
&\widetilde{g}^{(n+1)}(\aaa) + \lambda\widetilde{V}^{(n+1)}(\aaa)
\\
\label{gtilde}
\widetilde{g}^{(n+1)}(\aaa) \triangleq \,\, 
& (1-\phi^{(n+1)})\widetilde{g}^{(n)}(\aaa) 
+ \phi^{(n+1)}\widetilde{T}(\aaa)
\\
\label{Vtilde}
\nonumber
\widetilde{V}^{(n+1)}(\aaa) \triangleq \,\,
& V^{(n+1)} + ((\nabla V)^{(n+1)})^T(\aaa-\aaa^{(n+1)}) 
\\
&+ \tau||\aaa-\aaa^{(n+1)}||^2
\\
\nonumber
(\nabla V)^{(n+1)} = \,\,
& \pmb{1}-2\aaa^{(n+1)},
\end{align}
where $\phi^{(n+1)}$ is a weighting parameter and $\tau$ can be any positive constant. Here, the surrogate function $\widetilde{g}^{(n+1)}(\aaa)$ at iteration $n+1$ depends on the surrogate function $\widetilde{g}^{(n)}(\aaa)$ at iteration $n$ and the approximate functions $\widetilde{T}(\aaa)$ of $T(\aaa)$. $\widetilde{g}^{(n)}(\aaa)$ is approximately updated as
\begin{align}
\label{gtilden}
\widetilde{g}^{(n)}(\aaa) = g^{(n)} + (\nabla g)^{(n)} (\aaa-\aaa^{(n+1)}).
\end{align}
and $\widetilde{T}(\aaa)$ is updated as
\begin{align}
\label{Ttilde}
\nonumber
\widetilde{T}(\aaa) \triangleq\,\,
& T^{(n+1)} + ((\nabla T)^{(n+1)})^T(\aaa-\aaa^{(n+1)}) 
\\
& + \tau||\aaa-\aaa^{(n+1)}||^2.
\end{align}}

{Now, with \eqref{gtilden} and \eqref{Ttilde}, \eqref{gtilde}  becomes: 
\begin{align}
\label{gtilden+1}
\nonumber
\widetilde{g}^{(n+1)}(\aaa)  \triangleq \,\,
& g^{(n+1)} + ((\nabla g)^{(n+1)})^T(\aaa-\aaa^{(n+1)}) 
\\
& + \tau||\aaa-\aaa^{(n+1)}||^2,
\end{align}
where
\begin{align}
g^{(n+1)} =  \,\,
& (1-\phi^{(n+1)})g^{(n)}+ \phi^{(n+1)}T^{(n+1)} 
\\
\nonumber
(\nabla g)^{(n+1)} =  \,\,
&(1-\phi^{(n+1)})(\nabla g)^{(n)}
+ \phi^{(n+1)}(\nabla T)^{(n+1)} ,
\end{align}
with $g^{(0)}=0$, $(\nabla g)^{(0)} = \pmb{0}$. Here, 
\begin{align}
\nonumber
&(\nabla T)^{(n+1)} = 
\\
&q\frac{(\tilde{\ttt}_d + \tilde{\ttt}_c + \tilde{\ttt}_u)((\aaa^{(n+1})^T\pmb{1})) - \pmb{1}((\aaa^{(n+1)})^T(\tilde{\ttt}_d + \tilde{\ttt}_c + \tilde{\ttt}_u))} {((\aaa^{(n+1})^T\pmb{1}))^2}.
\end{align}
From \eqref{Vtilde} and \eqref{gtilden+1}, \eqref{LLtilden+1} can be written as
\begin{align}
\label{LLtilden}
\nonumber
\widetilde{\LL}^{(n+1)}(\aaa)  \triangleq \,\,
& \LL^{(n+1)} + ((\nabla \LL)^{(n+1)})^T(\aaa-\aaa^{(n+1)}) 
\\
& + \tau||\aaa-\aaa^{(n+1)}||^2,
\end{align}
where
\begin{align}
\LL^{(n+1)} = \,\,
& g^{(n+1)} + \lambda V^{(n+1)}
\\
(\nabla \LL)^{(n+1)} = \,\,
& (\nabla g)^{(n+1)} + \lambda (\nabla V)^{(n+1)}.
\end{align}
}


{At the large-scale coherence time or iteration $n+1$, problem \eqref{longP:appr:relax} is approximated by the following convex problem:
\vspace{-0mm}
\begin{align}\label{longP:appr:relax2}
\underset{\aaa\in\HHH}{\min} \,\,
&\widetilde{\LL}^{(n+1)}(\aaa).
\end{align}
}

\subsection{Solving the Overall Problem \eqref{mainP1:epi}}
\label{subsec:alg3}
Algorithm~\ref{alg:main} outlines the main steps to solve the overall problem \eqref{mainP1:epi} (hence \eqref{mainP1}). In the large-scale coherence time $n$ (i.e., iteration $n$), for a given random value of $\aaa^{(n+1)}\in\HHH$, the short-term subproblem \eqref{shortP} is solved by Algorithm~\ref{alg:2} after $I_S^{(n)}$ iterations to obtain a KKT solution.
This solution is then used to construct the approximate long-term master problem \eqref{longP:appr:relax2}. After solving \eqref{longP:appr:relax2} to obtain an optimal solution $(\aaa^*)^{(n+1)}$, we update $\aaa^{(n+2)}$ as
\begin{align}\label{a:update}
a_k^{(n+2)} = (1-\pi^{(n+1)})a_k^{(n+1)}+\pi^{(n+1)}(a_k^*)^{(n+1)}, \forall k,
\end{align}
where $\pi^{(n+1)}$ is a weighting parameter, and $\{\phi^{(n)},\pi^{(n)}\}$ are chosen to satisfy the following conditions \cite[Assumption 5]{liu18TSP}:
\begin{description}
\label{A1}
  \item[(B1):] $\phi^{(n)}\rightarrow 0$, $\frac{1}{\phi^{(n)}}\leq \OO (n^\varsigma)$ for $\varsigma\in(0,1)$, and $\sum_{n}(\phi^{(n)})^2<\infty$;
\label{A2}
  \item[(B2):] $\pi^{(n)}\rightarrow 0$, $\sum_{n}\pi^{(n)}=\infty$, $\sum_{n}(\pi^{(n)})^2<\infty$, and $\lim_{n\rightarrow\infty}\frac{\pi^{(n)}}{\phi^{(n)}}=0$.
\end{description}

Fig.~\ref{fig:time1} provides a top-level illustration of the operation of Algorithm~\ref{alg:main} within the context of CFmMIMO-based wireless FL. At the beginning, we run Algorithm~\ref{alg:main} based on the collected network information. Specifically, we solve the short-term subproblem \eqref{shortP} within a short-term optimization (STO) time block, and the long-term master problem \eqref{longP} within a long-term optimization (LTO) time block. 
Every two problems (21) and (22) are solved within a \textit{large-scale coherence time}.
Eventually, the converged UE selection solution $\aaa^{*}$ provided by Algorithm~\ref{alg:main} is used within the subsequent FL process as described in Sec.~\ref{subsec:CFmMIMO}. In the FL training process, one FL communication round and one STO time block are executed in one large-scale coherence time. Here, the values of transmit power and processing frequency obtained in the STO time block are not those obtained in the STO time block when executing Algorithm~\ref{alg:main}, but rather are computed by the same Algorithm~\ref{alg:2}. The detailed execution of one FL communication round is discussed in Section~\ref{subsec:CFmMIMO} and illustrated in Figure~\ref{fig:time2}. Once the network information changes, we will re-engage Algorithm~\ref{alg:main} to get a new solution, ready for FL to execute again.

\begin{algorithm}[!t]
\caption{Online successive convex approximation approach for solving the overall problem \eqref{mainP1:epi}}
\begin{algorithmic}[1]\label{alg:main}
\STATE \textbf{Initialization}: Set $n=0$, select a random $\aaa^{(n+1)}\in\HHH$
\REPEAT
\STATE Solve the short-term subproblem \eqref{shortP} to obtain its optimal solution $(\ETA^*,\f^*,\ZETA^*)$ by using Algorithm~\ref{alg:2}, and update $(\ETA^{(n+1)},\f^{(n+1)}\ZETA^{(n+1)})=
(\ETA^*,\f^*,\ZETA^*)$
\STATE Solve the approximate long-term master problem \eqref{longP:appr:relax2} of \eqref{longP} using CVX to obtain its optimal solution $(\aaa^*)^{(n+1)}$
\STATE Update $\aaa^{(n+2)}$ by \eqref{a:update} and $n=n+1$
\UNTIL{convergence}
\end{algorithmic}
\textbf{Output}: $\aaa^*=\aaa^{(n+1)}$
\end{algorithm}
\vspace{-0mm}

\begin{figure*}[t!]
	\centering
	\subfigure[Case (C1)]
	{\includegraphics[width=0.4\textwidth]{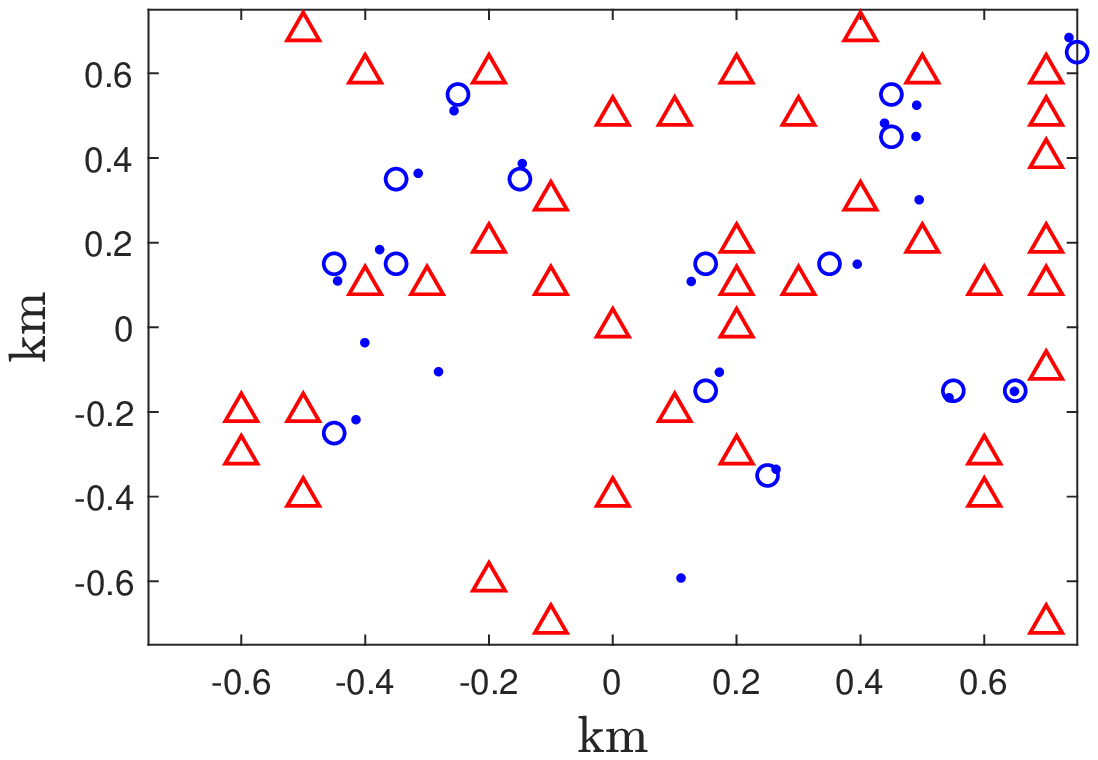}\label{subfig:C1}}
	\subfigure[Case (C2)]
	{\includegraphics[width=0.4\textwidth]{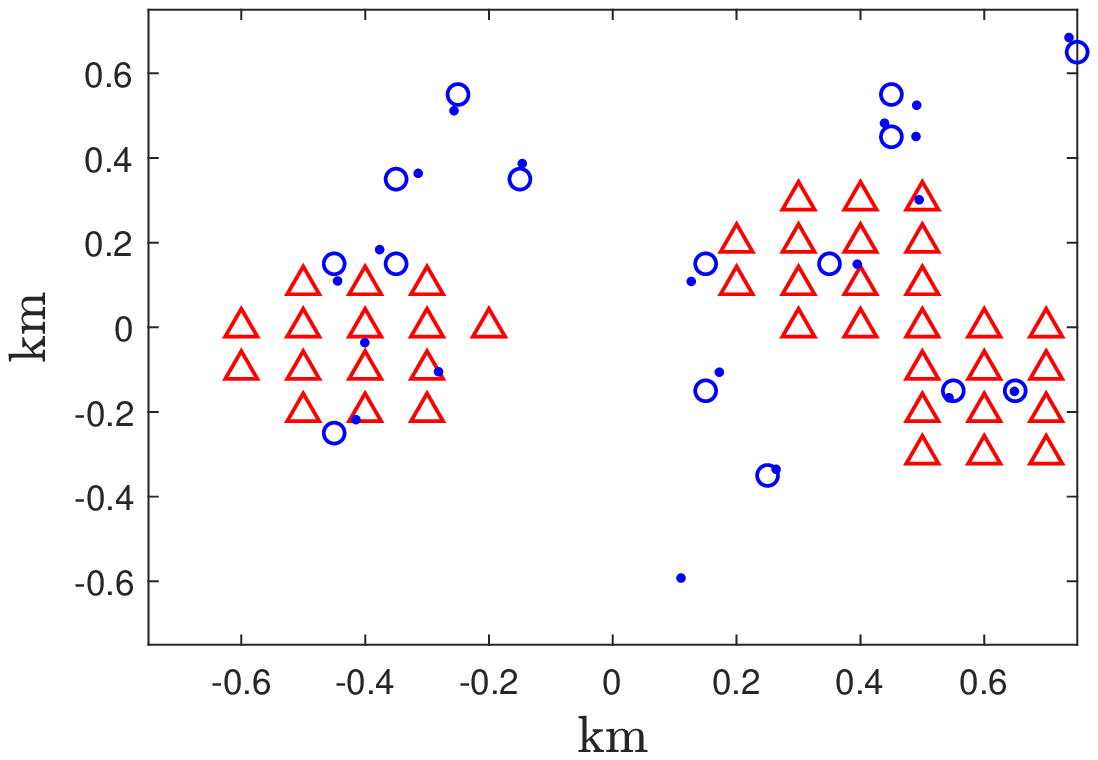}\label{subfig:C2}}
	\caption{{Examples of the network setup in two cases (C1) and (C2) with $40$ APs and $15$ UEs. The AP, UE, and fixed locations are represented as the markers in the shapes of triangles, circles, and dots, respectively.}}
	\label{Fig:system}
	\vspace{-2mm}
\end{figure*}

\subsection{Convergence and Complexity Analyses of Algorithm~\ref{alg:main}}

\begin{definition}
\label{stationarypoint}
A solution $(\aaa^*,\x^*)$ is called a stationary solution of problem \eqref{mainP1:epi} (or \eqref{mainP1}) if $\x^*$ is a KKT solution of the short-term subproblem \eqref{shortP} for $\aaa=\aaa^*$, and $\aaa^*$ is a KKT solution of the long-term master problem \eqref{longP} for $\x = \x^*$.
\end{definition}

\begin{proposition}
\label{Prop:mainP1}
Algorithm~\ref{alg:main} converges to a neighbourhood of the stationary solutions to problem \eqref{mainP1}.
\end{proposition}
\begin{proof}
  See Appendix~\ref{App:C}.
\end{proof}
\noindent
Theoretically, $I_S^{(n)}\rightarrow\infty$ and $\lambda\rightarrow\infty$ are required for Algorithm~\ref{alg:main} to converge to the stationary solutions of problem \eqref{mainP1}. When $I_S^{(n)}$ and $\lambda$ are finite, Algorithm~\ref{alg:main} converges to approximate stationary solutions of problem \eqref{mainP1}.

The computational complexities of solving \eqref{mainP:appr} at each iteration of Algorithm~\ref{alg:2} and solving \eqref{longP:appr:relax2} at each iteration of Algorithm~\ref{alg:main} are polynomial in the number of variables and constraints. In particular, \eqref{mainP:appr} can be transformed into an equivalent optimization problem that involves $N_{v,1}\triangleq(2MN+4N+3)$ real-valued scalar decision variables, $N_{l,1}\triangleq(2MN+M+4N+3)$ linear constraints and $N_{q,1}\triangleq(MN+5N)$ quadratic constraints. Therefore, \eqref{mainP:appr} requires a complexity of $\OO(\sqrt{N_{l,1}+N_{q,1}}[N_{v,1}+N_{l,1}+N_{q,1}]N_{v,1}^2)$ \cite{tam17TWC,Nemirovski96}. Problem \eqref{longP:appr:relax2} involves $N_{v,2}\triangleq N$ real-valued scalar decision variables, $N_{l,2}\triangleq(N+1)$ linear constraints. As such, \eqref{longP:appr:relax2} requires a complexity of $\OO(\sqrt{N_{l,2}}[N_{v,2}+N_{l,2}]N_{v,2}^2)$.
\vspace{-0mm}


\begin{figure*}[t!]
  \centering
  \vspace{-0mm}
  \subfigure[]
  {\includegraphics[width=0.4\textwidth]{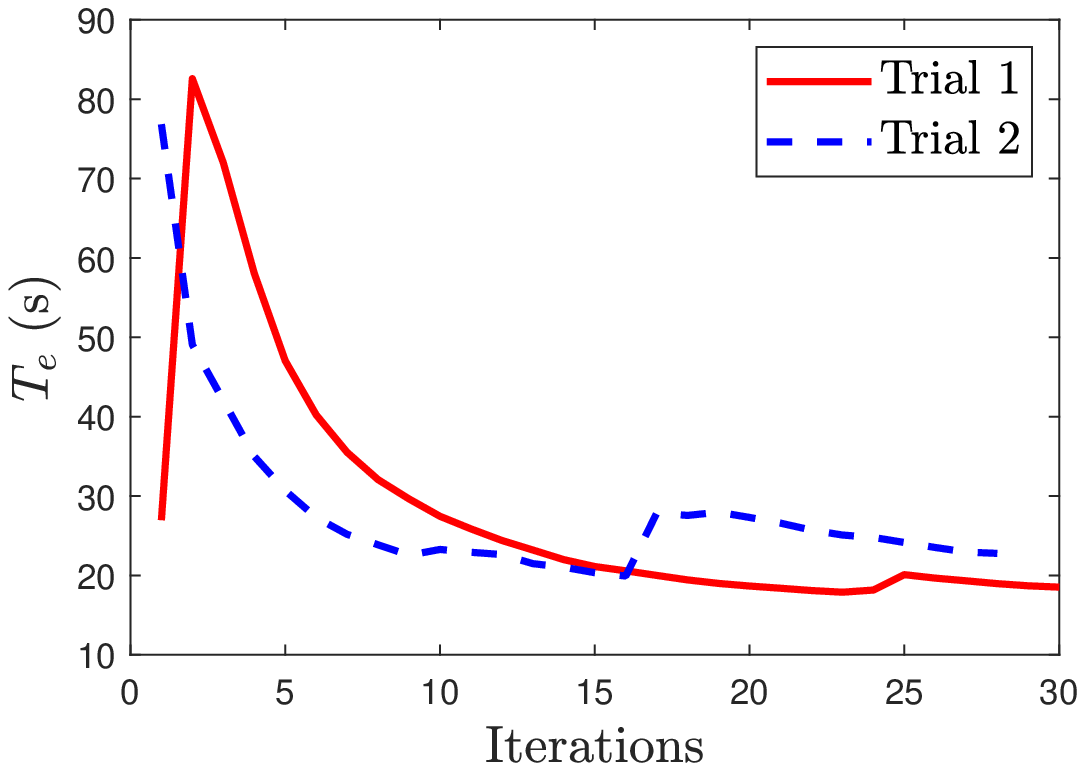}\label{subfig:1a}}
  \subfigure[]
  {\includegraphics[width=0.4\textwidth]{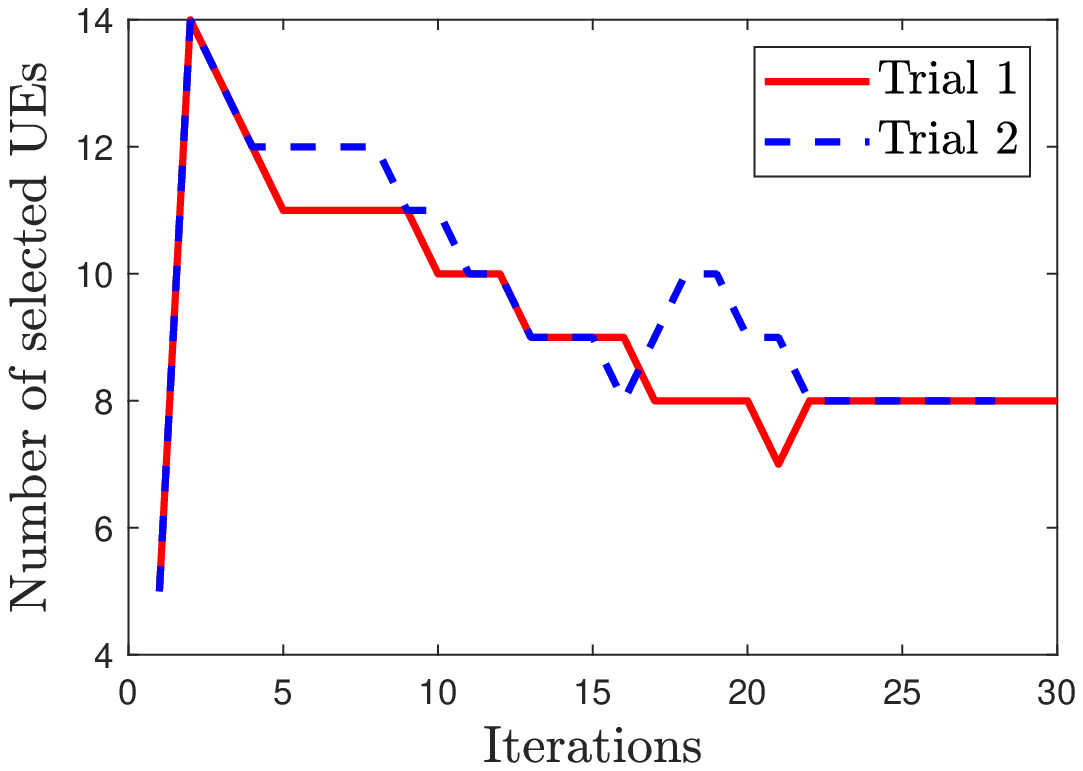}\label{subfig:1b}}
  \vspace{-2mm}
  \caption{{ The convergence process of Algorithm~\ref{alg:main}.
  In this example, we consider Case (C2) and set $M=40, N=15$, $N_{\text{QoL}}=5$, $D=1.5$ km.}}
  \label{Fig:1}
\end{figure*}

\vspace{-3mm}
\section{Numerical Examples}
\vspace{-0mm}
\label{sec:sim}
{
This section provides numerical results to analyze the effectiveness of the proposed Algorithm~\ref{alg:main} in minimizing the execution time of the considered standard FL process \cite{mcmahan17AISTATS,sai20ICML}. 
As previously discussed in Sec.~\ref{sec:Introd}, the ultimate question is how to find an optimal set of UEs to reduce the FL execution time without satisfying the test accuracy too much. To answer this question, we need to answer the fundamental questions (Q1) and (Q2). Assuming that $N_{\text{QoL}}$ is known in advance, our work focuses on answering question (Q2). It is shown in this section that the number of selected UEs in our example is always larger than  $N_{\text{QoL}}$. Therefore, the test accuracy obtained with this set of selected UEs is always acceptable. This means with the UE selection of our approach, the test accuracy of an FL process using real datasets is expected to be the same as that in \cite{mcmahan17AISTATS,sai20ICML}, and hence, not shown in this paper. On the other hand, if we try to make an ultimate analysis on the test accuracy of our UE selection scheme, we first have to set $N_{\text{QoL}}$ correctly, which leads us back to answering question (Q1). However, answering question (Q1) requires extensive efforts that are out of the scope of this work. Based on this observation, we leave the analysis of the test accuracy of an FL process with our UE selection scheme for future work.

}

\subsection{Network Setup with Non-Uniform UE/AP Distribution}
\label{subsec:netset}
\vspace{-0mm}
We consider a CFmMIMO network in a square of $D\times D$ km$^2$ whose edges are wrapped around to avoid the boundary effects. We examine the following two cases.
\begin{itemize}
    \item Case (C1): The UEs are more likely to stay near some fixed locations (e.g., coffee shops, restaurants) in the considered area; and the APs are uniformly distributed across the considered area.
    \item Case (C2): Both the APs and UEs are more likely to stay close to some fixed locations in the considered area.
\end{itemize}

\begin{figure*}[t!]
  \centering
  \vspace{-2mm}
  \subfigure[Case (C1) with $D=1.5$ km]
  {\includegraphics[width=0.4\textwidth]{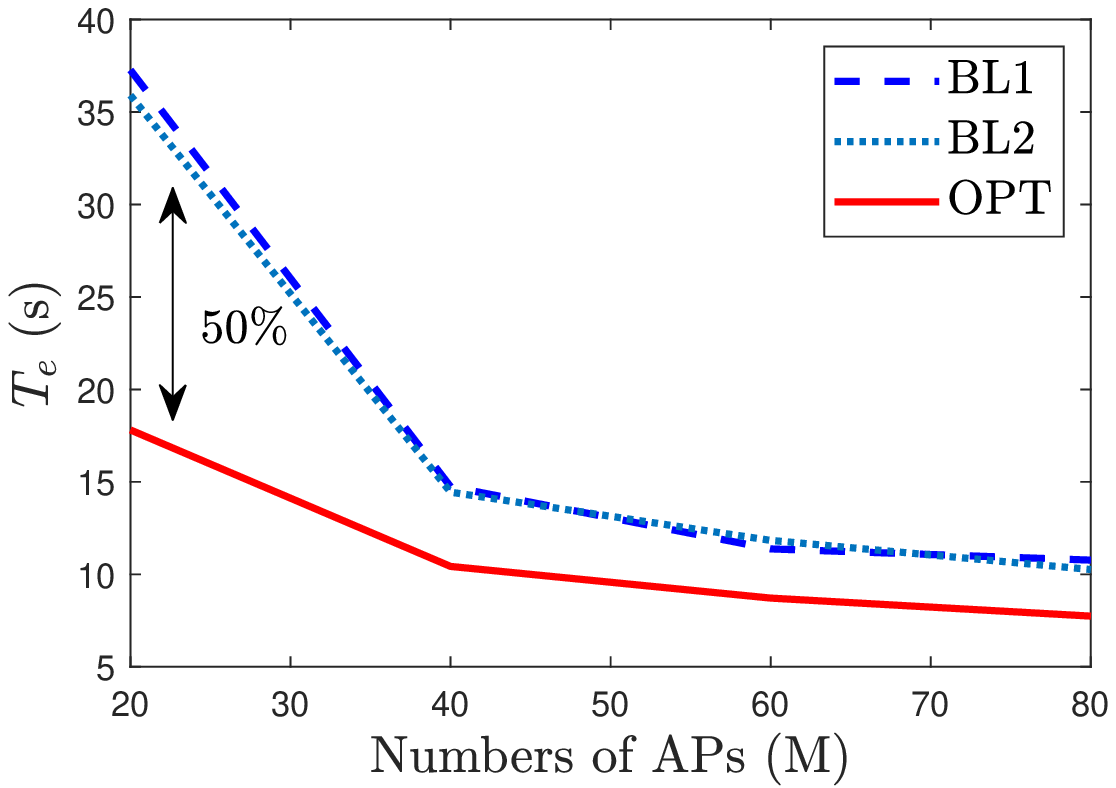}\label{subfig:2b}}
  \vspace{-2mm}
  \subfigure[Case (C2) with $D=1.5$ km]
  {\includegraphics[width=0.4\textwidth]{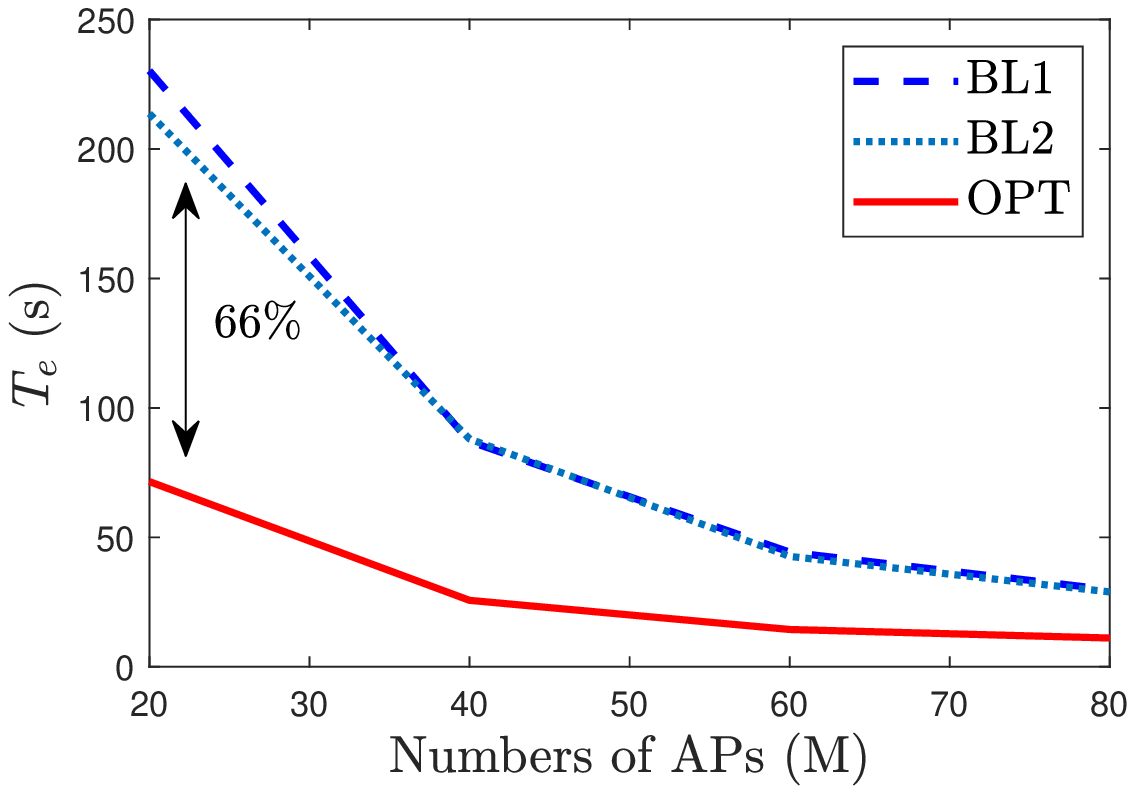}\label{subfig:2b}}
  \caption{{The effectiveness of the proposed approach. In these examples,  we set $N=15$ and $N_{\text{QoL}}=5$.}}
  \label{Fig:2}
\end{figure*}

\subsubsection{Modeling of Case (C1)}
{
First, a set of $x_p$ fixed locations are uniformly distributed over the network square. To generate UE locations, we create a grid that has $x_l$ vertical and $x_l$ horizontal lines. Then, $N$ UE locations are chosen as the points on this grid that are closest to $x_p$ fixed locations. To generate AP locations, we create a new grid that is the same as the UE location grid, but the lines of these two grids are interleaved. Finally, $M$ AP locations are uniformly chosen in the points on the latter grid. Here, we consider $x_l=15, x_p = 20$.}

\subsubsection{Modeling of Case (C2)}
{
Case (C2) is modeled in the same way for Case (C1), except that the $M$ AP locations are chosen from the points on the AP location grid that are closest to $x_{\text{AP}}$ fixed locations. Here, we choose $x_{\text{AP}}=3$, and the $x_{\text{AP}}$ fixed locations are uniformly selected out of $x_p$ fixed locations without replacement.
An example of the network setup is shown in Fig.~\ref{Fig:system}.}

\subsubsection{Setup for each network realization}
{
Since each FL communication round happens in one large-scale coherence time (in the order of seconds), the total execution time of an FL process is expected to be around several minutes. Therefore, we assume that the UEs only move around their current locations during the FL process. Here, in each communication round of an FL process, we let each UE move within a circle of radius $5$ m around its current location, while the AP locations remain unchanged.} 

\vspace{-3mm}
\subsection{Parameter Settings}
We model large-scale fading coefficients $\beta_{mk}$ as \cite{emil20TWC}: 
\begin{align}\label{fading:large}
\beta_{mk} = 10^{\frac{\text{PL}_{mk}^d}{10}}10^{\frac{F_{mk}}{10}},
\end{align}
where $10^{\frac{\text{PL}_{mk}^d}{10}}$ represents the path loss, and $10^{\frac{F_{mk}}{10}}$ represents the shadowing effect with $F_{mk}\in\NN(0,4^2)$ (in dB).  Here, $\text{PL}_{mk}^d$ (in dB) is given by  \cite{emil20TWC}
\begin{align}\label{PL:model}
\text{PL}_{mk}^d = -30.5-36.7\log_{10}\left(\frac{d_{mk}}{1\,\text{m}}\right),
\end{align}
and the correlation among the shadowing terms from the AP $m, \forall m\in\MM$ to different UEs $k,\ell\in\NN$ is expressed as:
\begin{align}\label{corr:shadowing}
\EEE\{F_{mk}F_{j\ell}\} \triangleq
\begin{cases}
  4^22^{-\delta_{k\ell}/9\,\text{m}},& \text{if $j=m$}\\
  0, & \mbox{otherwise},
\end{cases}, \forall j\in\MM,
\end{align}
where $\delta_{k\ell}$ is the physical distance between UEs $k$ and $\ell$.

For channel estimation, we use a random pilot assignment scheme. Specifically, the pilot of each user is randomly chosen from a predefined set of $\tau_t=N$ orthogonal pilot sequences, each having a length of $\tau_t$ samples. We set $\tau_c\!=\!200$ samples, $S_d\!=\!S_u\!=\!5$ MB, noise power $\sigma_0^2\!=\!-92$ dBm, $L=5$, $f_{\max}=3 \times 10^9$ cycles/s, $D_k = D = 5\times 10^6$ samples, $c_{k} = 20$ cycles/samples \cite{tran19INFOCOM}, for all $k$, $\alpha=2\times 10^{-29}$. We choose 
$q=90$. Let $\tilde{\rho}_d=1$ W, $\tilde{\rho}_u=0.2$ W, $\tilde{\rho}_t=0.2$ W be the maximum transmit powers of the APs, UEs, and uplink pilot sequences, respectively.
Here, $\rho_d$, $\rho_u$ and $\rho_t$ are the normalized values of $\tilde{\rho}_d$, $\tilde{\rho}_u$ and $\tilde{\rho}_t$ with respect to the noise power. We set $\pi^{(n)} = \frac{1000}{1000+n}$ and $\phi^{(n)} = \frac{1}{n^{9/10}}$ to satisfy conditions (B1) and (B2) in Section~\ref{subsec:alg3}. Finally, we choose $\lambda=1$.

\subsection{Results and Discussions}
\label{discuss}
\vspace{-0mm}
\subsubsection{Effectiveness of Algorithm~\ref{alg:main}}
\label{discuss:effectiveness}
First, we evaluate the convergence behavior of the proposed Algorithm~\ref{alg:main}. As seen from Fig.~\ref{Fig:1}, Algorithm~\ref{alg:main} converges within $30$ iterations for an arbitrary network realization. Note that each iteration of Algorithm~\ref{alg:main} involves solving simple convex programs \eqref{mainP:appr} and \eqref{longP:appr:relax2}. It is therefore expected that  Algorithm~\ref{alg:main} has a low computational complexity.

\begin{figure*}[t!]
	\centering
	\subfigure[]
	{\includegraphics[width=0.4\textwidth]{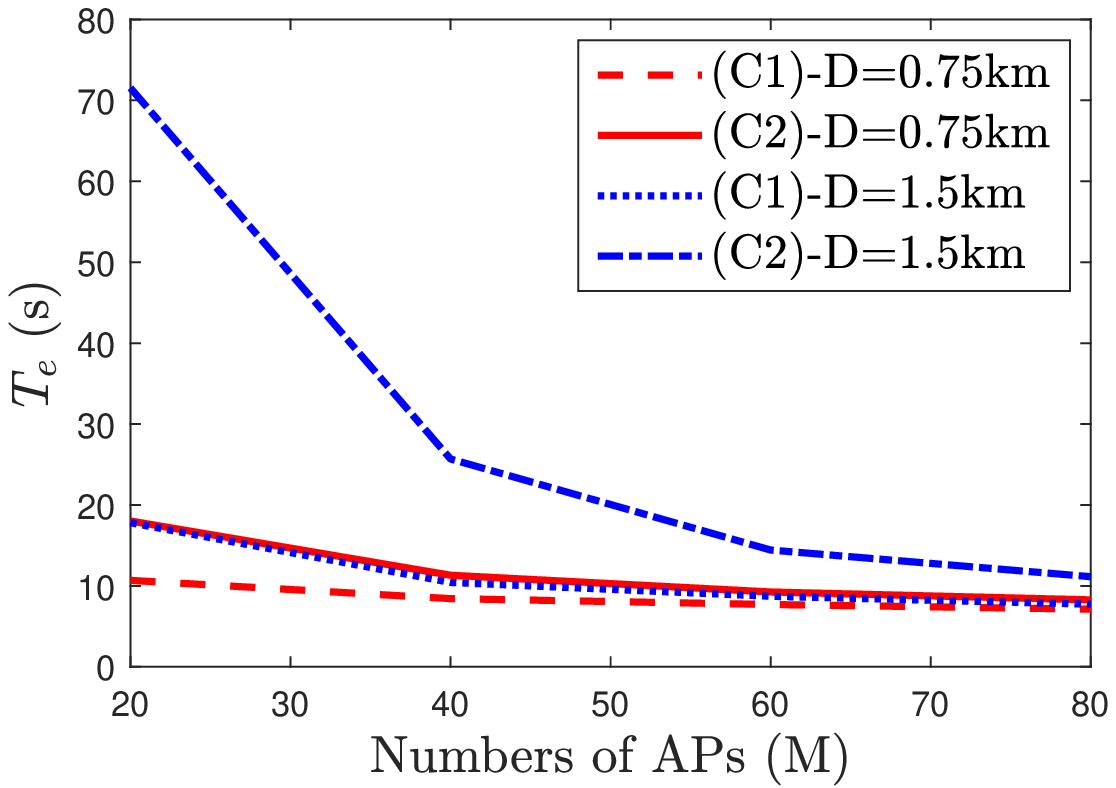}\label{subfig:4a}}
	\subfigure[]
	{\includegraphics[width=0.4\textwidth]{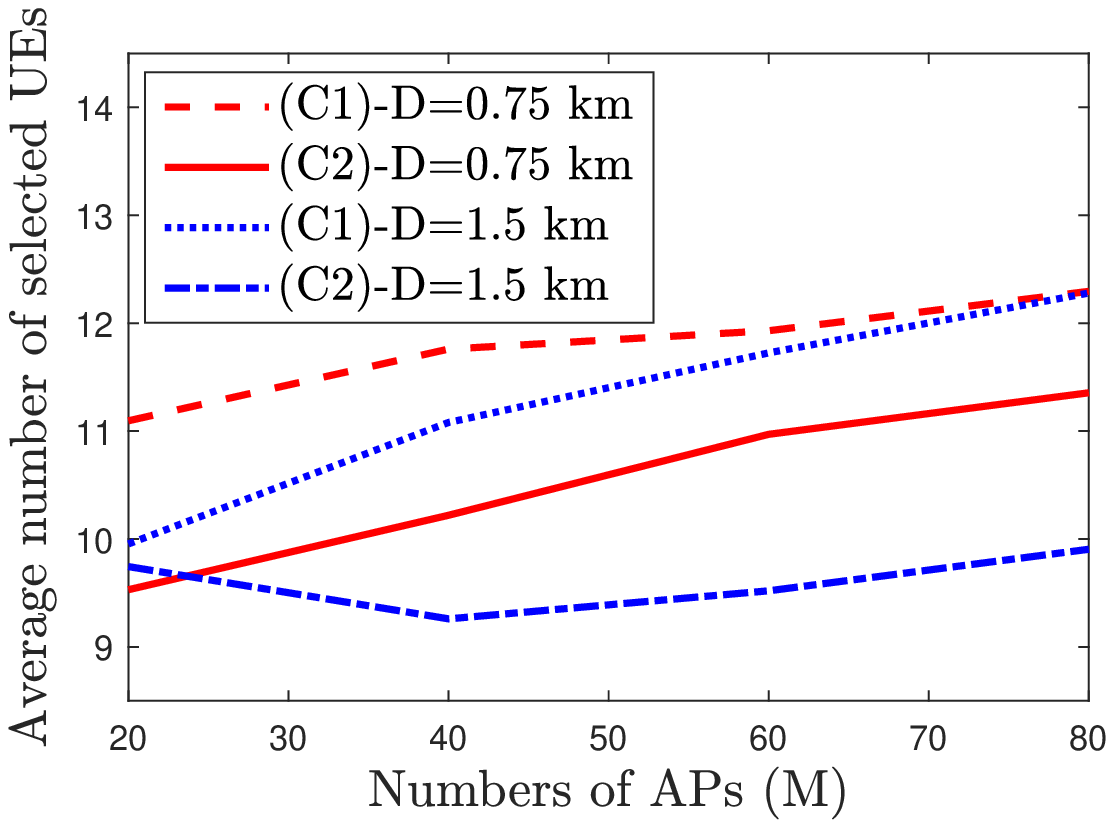}\label{subfig:4b}}
	\vspace{-0mm}
	\caption{{Impact of the number of APs on: (a) FL execution time; (b) number of selected UEs (b). In this example, we set $N=15$ and $N_{\text{QoL}}=5$.}}
	\label{Fig:4}
\end{figure*}

\begin{figure*}[t!]
	\centering
	\subfigure[]
	{\includegraphics[width=0.4\textwidth]{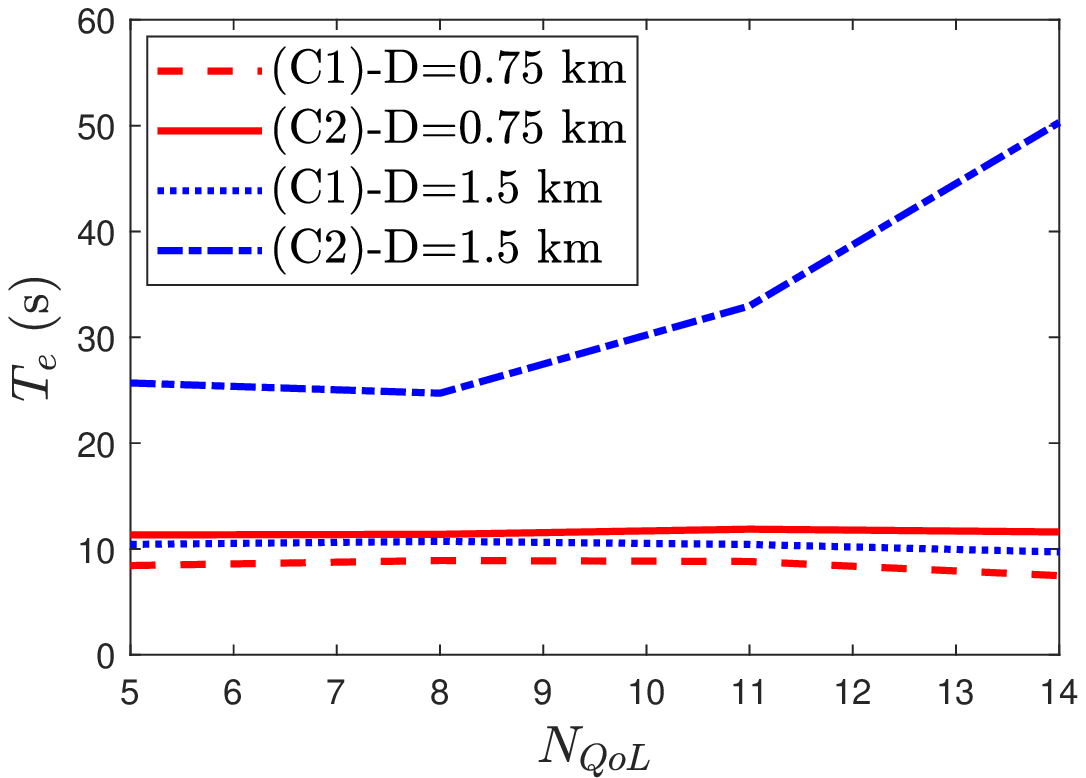}\label{subfig:5a}}
	\subfigure[]
	{\includegraphics[width=0.4\textwidth]{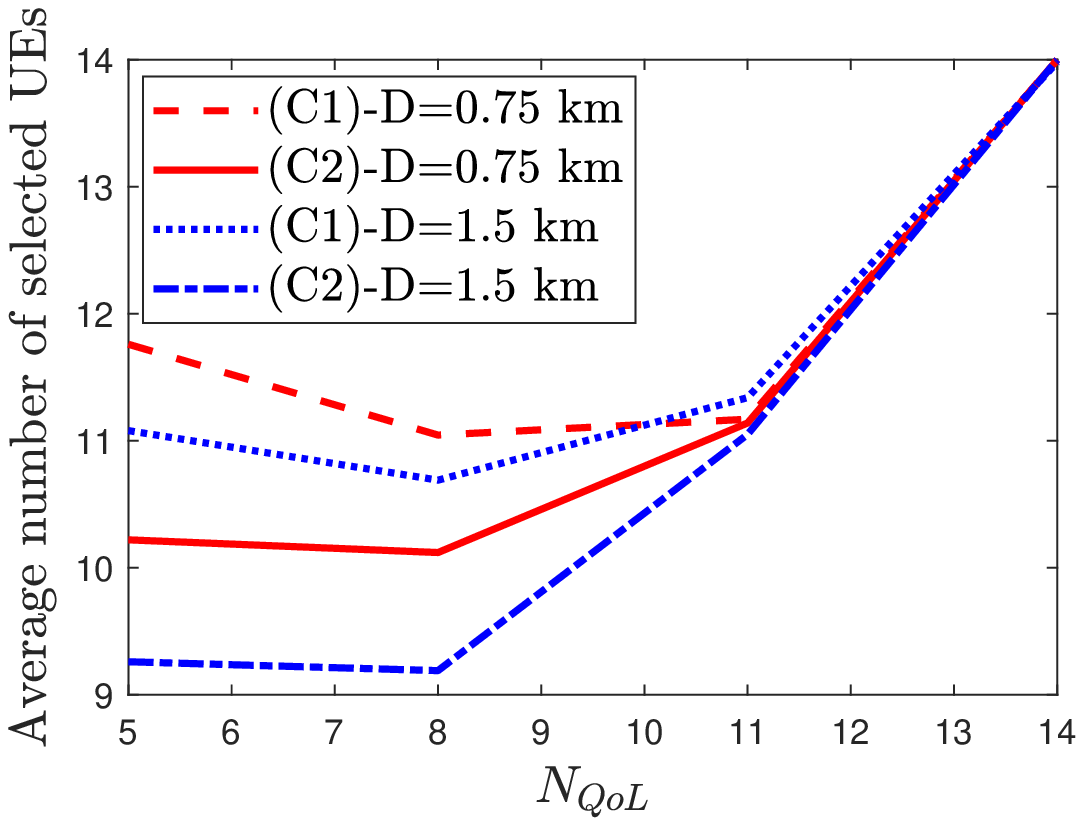}\label{subfig:5b}}
	\vspace{-0mm}
	\caption{{Impact of the threshold $N_{\text{QoL}}$ on: (a) FL execution time; (b) number of UEs being selected. In this example, we set $M=40$, $N=15$.}}
	\label{Fig:5}
\end{figure*}

{Next, we compare Algorithm~\ref{alg:main} (denoted by \textbf{OPT} in the figures) with the following baseline schemes:}
\vspace{-0mm}
\begin{itemize}
\item {
Baseline 1 (\textbf{BL1}): Let $\widehat{N}$ be an integer uniformly drawn from the interval $[N_{\text{QoL}},N]$. Then, before an FL process is executed, we select participating UEs for this FL process by uniformly selecting $\widehat{N}\leq N$ UEs without replacement.
The UE selection result of this scheme is presented by a vector $\aaa_{\text{BL1}}$. The FL execution time by \textbf{BL1} is thus $T_e(\aaa_{\text{BL1}},\ETA,\f,\ZETA)= G(\aaa_{\text{BL1}}) \EEE\{T_o(\aaa_{\text{BL1}},\ETA,\f,\ZETA)\}$. 
Since $\aaa_{\text{BL1}}$ is known, solving the problem of minimizing $T_e(\aaa_{\text{BL1}},\ETA,\f,\ZETA)$ only requires solving the short-term subproblem \eqref{shortP} of optimizing $(\ETA,\f,\ZETA)$ to minimize the execution time of one FL communication round. Here, \eqref{shortP} is solved by using Algorithm~\ref{alg:2} for a given $\aaa_{\text{BL1}}$.
}
\item Baseline 2 (\textbf{BL2}):
{
In this baseline, we let all $N$ UEs participate in an FL \emph{process} but only $K\leq N$ UEs participate in each FL \emph{communication round}. Here, the integer $K$ is uniformly drawn from the interval $[N_{\text{QoL}},N]$.
Then, $K$ UEs are uniformly chosen out of $N$ original UEs without replacement.
The UE selection result of this scheme in each FL communication round is represented by a vector $\aaa_{\text{BL2}}$.
Note that \textbf{BL2} is identical to the opportunistic UE sampling scheme proposed in \cite{mcmahan17AISTATS,sai20ICML}.
Since \textbf{BL2} performs UE selection in each FL communication round, the number of FL communication rounds is $\widetilde{G} = \frac{q}{K} + \tilde{q}\big(1-\frac{K}{N}\big)$\cite[Theorem 1]{sai20ICML}, where $q, \tilde{q}$ depend on the specific characteristics of the FL learning problems and are assumed known in advance. Here, since there is no UE sampling in each FL communication round in our proposed approach, all $\widetilde{N}=\sum_{k\in\NN}a_k$ selected UEs participate in an FL process, i.e., $K=\widetilde{N}$. Therefore, the number of FL communication rounds are $\frac{q}{\widetilde{N}}$ as in \eqref{G}. 
In this work, we choose $\tilde{q} = q$ for simplicity.
Then, the execution time of an FL process using  \textbf{BL2} is measured by $T_e(\aaa_{\text{BL2}},\ETA,\f,\ZETA) = \widetilde{G}\EEE\{T_o(\aaa_{\text{BL2}},\ETA,\f,\ZETA)\}$. 
Since $\aaa_{\text{BL2}}$ is known in each FL communication round, solving the problem of minimizing $T_e(\aaa_{\text{BL2}},\ETA,\f,\ZETA)$ only requires solving the short-term subproblem \eqref{shortP} of optimizing $(\ETA,\f,\ZETA)$ to minimize the execution time of one FL communication round. Here, \eqref{shortP} is solved by using Algorithm~\ref{alg:2} for a given $\aaa_{\text{BL2}}$.
}
\end{itemize}

Fig.~\ref{Fig:2} compares the total execution time of an FL process by all the considered schemes. As seen, our \textbf{OPT} scheme is the best performer. 
In particular, while \textbf{BL1} and \textbf{BL2} perform quite similarly, \textbf{OPT} cuts the execution time by a substantial amount, e.g., by up to $50\%$ in Case (C1) and $66\%$ in Case (C2) with $M=20$ and $D=1.5$ km. 
These results show the significant advantage of an optimal UE selection over heuristic UE selections.

{
Fig.~\ref{Fig:2}  also shows the importance of optimal UE selection for reducing the FL execution time when the AP density defined as the number of APs over a geographical area is moderately low. 
Specifically, in both Cases (C1) and (C2) with a large value of $D$ or a low value of $M$, the reduction in the FL execution time by \textbf{OPT} is at least $44\%$.
This is reasonable because in these cases, there is a high probability of having UEs with unfavorable links. 
This leads to a significantly low execution time of one FL communication time, and hence, the whole FL execution time. 
}

\subsubsection{Impact of the number of APs on the number of selected UEs}
{
Fig.~\ref{Fig:4} shows that 
a larger number of APs corresponds to a larger number of UEs being selected. Here, as the AP-UE distances are smaller, there are potentially more UEs with favorable links hence being selected. 
The only time the number of UEs being selected decreases is in Case (C2) with $M=\{20,40\}$ and $D=1.5$ km. In this case, there are a high probability of having UEs with unfavorable links, which leads to the two largest execution times as shown in Fig.~\ref{Fig:4}(a).
Therefore, the numbers of selected UEs still need to be reduced to shorten the execution time of one FL communication, and hence, the whole FL execution time.}

\subsubsection{Impact of $N_{\text{QoL}}$ on the execution time of an FL process}
{
Fig.~\ref{Fig:5}(a) shows that 
increasing $N_{\text{QoL}}$ leads to a dramatic increase in the FL execution time in a network that has a significantly low density of APs and non-uniformly distributed AP locations. This network is presented in Case (C2) with $D=1.5$ km, where there is a high probability of having UEs with unfavorable links. In this case, for a larger value of $N_{\text{QoL}}$, more UEs are required to participate in an FL process as shown in Fig.~\ref{Fig:5}(b). While this reduces the number of FL communication rounds, a stronger inter-user interference is also resulted. In our example, such a reduction in the number of communication round is not sufficient to compensate for the increase in the execution time of each FL communication round.} 

{
Fig.~\ref{Fig:5}(a) also shows that CFmMIMO networks with a high density of APs potentially provides low-latency FL services for everyone. When there are fewer UEs with unfavorable links as in Case (C1) with $D=\{0.75, 1.5\}$ km and Case (C2) with $D=0.75$ km, the FL execution times are nearly the same when $N_{\text{QoL}}$ is increased.} 


\section{Conclusion}
\label{sec:con}
{
In this work, we have proposed a novel approach that jointly designs UE selection, transmit power, and processing frequency to minimize the execution time of an FL process in CFmMIMO networks. We formulate a mixed-integer mixed-timescale stochastic nonconvex problem under practical requirements on the maximum transmit powers and the minimum number of selected UEs to guarantee quality of learning. Utilizing online successive convex approximation, we have successfully developed a novel algorithm to solve the formulated problem. The proposed algorithm has been proved to converge to the neighbourhood of stationary points. Numerical results have showed that our approach significantly reduces the FL execution time over the baseline schemes.}


%
\vspace{-0mm}

\appendices

\section{Proof of Proposition~\ref{Prop:mainP1}}
\label{App:C}
The proof involves two steps. The first step is to prove that the solution $\x^*$ obtained from Algorithm~\ref{alg:main} is a KKT solution of the short-term subproblem \eqref{shortP}. This proof has already been provided in Proposition~\ref{shortP:Prop}. The second step is to prove that the solution $\aaa^{(n+1)}$ obtained from Algorithm~\ref{alg:main} is a KKT solution of the long-term master problem \eqref{longP}. The details of this proof are as follows. 

It can be confirmed that problem \eqref{mainP1:epi} satisfies the conditions of Assumption 1 on the main problem in the general framework \cite{liu18TSP}. It is worth noting that we do not verify Assumption 1-5) and Assumption 1-6) of \cite{liu18TSP} for the following reasons. Assumption 1-5) on Mangasarian-Fromovitz constraint qualification is used to ensure the existence of KKT solutions of the short-term subproblem \eqref{shortP}. In this work, since Proposition~\ref{shortP:Prop} shows that a KKT solution of \eqref{shortP} can be obtained by Algorithm~\ref{alg:2}, this assumption is unnecessary.
Assumption 1-6) is used to guarantee convergence to an exact stationary point of the short-term subproblem \eqref{shortP}. However, \cite{liu18TSP} confirms that Assumption 1-6) can be removed when we allow an approximate convergence by solving the short-term subproblem \eqref{shortP} with a finite number of iterations.

From the definitions of $\widetilde{R}_{d,k}(\vv)$, and $\widetilde{R}_{u,k}(\uu)$ in  \eqref{hd:apprx} and \eqref{hu:apprx}, it can be verified that $\widetilde{R}_{d,k}(\vv)$ and $\widetilde{R}_{u,k}(\uu)$ have the following properties:
\begin{itemize}
  \item  $\widetilde{R}_{d,k}(\vv^{(n)})=R_{d,k}(\vv^{(n)})$, $\widetilde{R}_{u,k}(\uu^{(n)})={h}_{u,k}(\uu^{(n)})$,
        $\nabla\widetilde{R}_{d,k}(\vv^{(n)})=\nabla R_{d,k}(\vv^{(n)})$, $\nabla\widetilde{R}_{u,k}(\uu^{(n)})=\nabla{R}_{u,k}(\uu^{(n)})$;
  \item $-\widetilde{R}_{d,k}(\vv)$, and  $-\widetilde{R}_{u,k}(\uu)$ are strongly convex;
  \item   $\widetilde{R}_{d,k}(\vv,\vv^{(n)})$ and $\widetilde{R}_{u,k}(\uu,\uu^{(n)})$ are Lipschitz continuous in both $\vv,\vv^{(n)}$ and both $\uu,\uu^{(n)}$, respectively.
\end{itemize}
Algorithm~\ref{alg:2} thus satisfies all the conditions of Assumption 2 on the short-term algorithm within the general framework \cite{liu18TSP}.
Since $\{\phi^{(n)},\pi^{(n)}\}$ are chosen to satisfy conditions (B1) and (B2) in Sec.~\ref{subsec:alg3}, they satisfy all the conditions of Assumption 5 in \cite{liu18TSP}. When Assumptions 1, 2 and 5 in \cite{liu18TSP} is satisfied, it is confirmed by \cite[Corollary 1]{liu18TSP} that the surrogate function $\widetilde{\LL}(\aaa)$ in \eqref{longP:appr:relax2} satisfies the Assumptions 3 and 4 in \cite{liu18TSP} on the properties and asymptotic consistency of surrogate functions.

Since Assumptions 1-5 in \cite{liu18TSP} are all satisfied, it follows from \cite[Lemma 1]{liu18TSP} that:
\renewcommand{\labelenumi}{(\roman{enumi})}
\begin{enumerate}
  \item The sequence $\{(\aaa^{(n+1)},(\aaa^*)^{(n+1)})\}_{n=1}^{\infty}$ generated over iterations of Algorithm~\ref{alg:main} has the following property.
\begin{align}\label{aclose}
  \lim_{n\rightarrow\infty}||\aaa^{(n+1)}-(\aaa^*)^{(n+1)}||=0.
\end{align}
  \item Let $\aaa_{\star}$ be a limit point of a subsequence $\{\aaa^{(n+1)_j}\}_{j=1}^{\infty}$ and
\begin{align}
\label{lim1}
\lim_{j\rightarrow\infty} |\widetilde{\LL}(\aaa^{(n+1)_j})-\widetilde{\LL}(\aaa_{\star})|=0,
\\
\label{lim2}
\lim_{j\rightarrow\infty} |\nabla\widetilde{\LL}(\aaa^{(n+1)_j})-\nabla \LL(\aaa_{\star})|=0.
\end{align}
\end{enumerate}
Without loss of generality, we assume that $\aaa^{(n+1)}\to\aaa_\star$ as $n\to\infty$. Then, \eqref{lim1} and \eqref{lim2} imply that
\begin{align}
\label{lim3}
\lim_{n\rightarrow\infty} |\widetilde{\LL}(\aaa^{(n+1)})-\LL(\aaa^{(n+1)})|=0,
\\
\label{nablagclose}
\lim_{n\rightarrow\infty} |\nabla\widetilde{\LL}(\aaa^{(n+1)})-\nabla \LL(\aaa^{(n+1)})|=0.
\end{align}

It can be seen that there always exists one interior point in $\HHH$. Therefore, the convex problem \eqref{longP:appr:relax2} satisfies the Slater's constraint qualification condition. Its optimal solution $(\aaa^*)^{(n+1)}$ is thus a KKT solution to \eqref{longP:appr:relax}, and hence, \eqref{longP} when $\lambda\rightarrow \infty$ (see Proposition~\ref{Prop:LongP1}), i.e.,
\begin{subequations}\label{KKT}
\begin{align}
&\nabla {\LL}((\aaa^*)^{(n+1)})
+ \sum_{j=1}^r\nu_j\nabla \delta_j((\aaa^*)^{(n+1)})
= 0,
\\
&\nu_j\delta_j((\aaa^*)^{(n+1)})=0, \forall j\in\{1,...,q\},
\end{align}
\end{subequations}
where $\delta_j(\aaa),\forall j\in\{1,\dots,q\}$ represent the functions in the constraints \eqref{tildeN}, \eqref{suma}
and \eqref{arelax}. It follows from \eqref{aclose} and \eqref{nablagclose} that the gap between $\aaa^{(n+1)}$ and $(\aaa^*)^{(n+1)}$ and that between $\nabla\widetilde{\LL}(\aaa^{(n+1)})$ and $\nabla\LL(\aaa^{(n+1)})$ converge to zero as $n\to\infty$. Therefore, \eqref{KKT} implies
\begin{subequations}
\begin{align}
&\nabla \LL(\aaa^{(n+1)})
+ \sum_{j=1}^r\nu_j\nabla \delta_j(\aaa^{(n+1)})
= 0,
\\
&\nu_j\delta_j(\aaa^{(n+1)})=0, \forall j\in\{1,...,r\},
\end{align}
\end{subequations}
which means $\aaa^{(n+1)}$ is a KKT solution of the long-term master problem \eqref{longP}.

As such, the convergence of Algorithm~\ref{alg:main} to a stationary point of problem \eqref{mainP1:epi} in the sense of Definition~\ref{stationarypoint} are guaranteed if the numbers of iterations of Algorithms~\ref{alg:2} and~\ref{alg:main} are infinity, i.e., $I_S^{(n)}\rightarrow \infty$, $I_L\rightarrow \infty$, and $\lambda\to\infty$. In practice, it is acceptable to choose finite $\{I_S^{(n)}\}_{n\in\{1,\dots,I_L\}}$, $I_L$, and $\lambda$ for an approximate convergence. Therefore, Algorithm~\ref{alg:main} is guaranteed to converge to the neighbourhood of the stationary solutions of problem \eqref{mainP1:epi}, and hence, \eqref{mainP1}.



\ifCLASSOPTIONcaptionsoff
  \newpage
\fi

\bibliographystyle{IEEEtran}
\bibliography{IEEEabrv,newidea2020}
\end{document}